\documentclass[final,leqno,onefignum,onetabnum]{siamltex1213}

\usepackage{amsmath,amssymb} %,amsthm}
\usepackage{enumerate}
\usepackage{natbib}
\usepackage{algorithm,algorithmic}
\usepackage{xcolor,xfrac}	
\usepackage{array, multirow}
\usepackage{bbm}
\def\bb{\textbf{b}}
\def\bc{\textbf{c}}

\def\bl{\textbf{l}}

\def\bq{\textbf{q}}

\def\bt{\textbf{t}}
\def\bu{\textbf{u}}
\def\bv{\textbf{v}}
\def\bw{\textbf{w}}
\def\bx{\textbf{x}}  %{\mbox{\boldmath $\lambda$}}
\def\by{\textbf{y}}
\def\bz{\textbf{z}}

\def\bD{\textbf{D}}

\def\bL{\textbf{L}}

\def\bR{\textbf{R}}
\def\bS{\textbf{S}}

\def\bZ{\textbf{Z}}
\def\bo{\textbf{0}}

\def\bfell{{\boldsymbol{\ell}}}

\def\bfgamma{{\boldsymbol{\gamma}}}

\def\bfomega{{\boldsymbol{\omega}}}

\def\bfmu{{\boldsymbol{\mu}}}

\def\bfSigma{{\boldsymbol{\Sigma}}}
\def\bfsigma{{\boldsymbol{\sigma}}}
\def\bfpi{{\boldsymbol{\pi}}}

\def\bftheta{{\boldsymbol{\theta}}}

\def\bfxi{{\boldsymbol{\xi}}}
\def\bfbeta{{\boldsymbol{\beta}}}
\def\bfrho{{\boldsymbol{\rho}}}
\def\bfzeta{{\boldsymbol{\zeta}}}

\def\cL{\mathcal{L}}

\def\cO{\mathcal{O}}

\def\mE{\mathbb{E}}

\def\mR{\mathbb{R}}

\def\eq{\[}
\def\en{\]}

\def\smskip{\smallskip}
 
\def\texitem#1{\par\smskip\noindent\hangindent 25pt
               \hbox to 25pt {\hss #1 ~}\ignorespaces}

\def\abs#1{\left|#1\right|}
\def\norm#1{\left\|#1\right\|}

\newcommand{\BEAS}{\begin{eqnarray*}}
\newcommand{\EEAS}{\end{eqnarray*}}
\newcommand{\BEA}{\begin{eqnarray}}
\newcommand{\EEA}{\end{eqnarray}}
\newcommand{\BEQ}{\begin{eqnarray}}
\newcommand{\EEQ}{\end{eqnarray}}
\newcommand{\BIT}{\begin{itemize}}
\newcommand{\EIT}{\end{itemize}}
\newcommand{\BNUM}{\begin{enumerate}}
\newcommand{\ENUM}{\end{enumerate}}

\newcommand{\BA}{\begin{array}}
\newcommand{\EA}{\end{array}}

\newcommand{\ones}{\mathbf 1}

\newcommand{\reals}{{\mathbb{R}}}

\newcommand{\argmin}{\mathop{\rm argmin}}
\newcommand{\argmax}{\mathop{\rm argmax}}

\newcommand{\ES}[2]{\text{ES}_{#1}(#2)}
\def\algo{\textsc{SpecRiskAllocate}}
\def\FISTA{\textsc{FISTA}}

\newcolumntype{H}{>{\setbox0=\hbox\bgroup}c<{\egroup}@{}}

\title{Portfolio selection with multiple spectral risk constraints} 

\author{Carlos Abad\thanks{IEOR, Columbia
    University. \email{ca2446@columbia.edu}} 
  \and
  Garud Iyengar\thanks{IEOR, Columbia
    University. % \email{garud@ieor.columbia.edu}}
  Supported in part by NSF
  grants DMS-1016571, DOE grant DE-FG02-08ER25856 and ONR grant
  N000140310514}}

\begin{document}
\maketitle

\begin{abstract}
  We propose an iterative gradient-based algorithm to efficiently solve the portfolio
  selection problem with multiple spectral risk constraints. Since the
  conditional value at risk~(CVaR) is a special case of the spectral
  risk measure, our algorithm solves portfolio selection problems
  with multiple CVaR constraints. 
  %% Garud
  In each step, the algorithm solves
  very simple separable convex quadratic programs; hence, we show that the
  spectral risk constrained portfolio selection problem 
  can be solved using the 
  technology developed for solving mean-variance problems.
  %%%%%%%%%%% 
  The algorithm
  extends to the case where the objective is a weighted sum of the
  mean return and either a weighted combination or the maximum of a set of spectral risk
  measures. We
  report numerical results that show that our proposed algorithm is
  very efficient; it is at least one order of magnitude faster than
  the state-of-the-art general purpose solver for all practical
  instances.  
  % Moreover, we show how the efficiency of the proposed algorithm allows
  %% Garud
  One can leverage this efficiency to be robust against model risk by
  including  
  % portfolio managers to include
  constraints with respect to several different risk
  models. % in order to select portfolios that are robust with respect to
  % uncertainty in risk parameters. 
\end{abstract}



\section{Introduction}\label{intro}
%% Garud
Portfolio selection %problem % can be broadly defined as that of
is concerned with
distributing a given % amount of
capital over a finite number of
investment opportunities in order to maximize ``return'' while managing
``risk''. Although, the benefits of diversification to manage ``risk'' had
been long known, \citet{Markowitz1952} was the first to propose a mathematical model
for the portfolio optimization problem, representing % . In his model, Markowitz represented
``return'' by the expected return of the portfolio, and ``risk'' by the
variance in the return of the portfolio.
It has been observed that variance is a good measure of risk \emph{only} if 
the returns are \emph{elliptically} distributed.
%The goal of portfolio selection is to distribute a fixed amount of
%capital over a given set of investment opportunities to maximize
%``return'' while managing ``risk''. Although the benefits of
%diversifying were well-known, the first mathematical model for
%portfolio selection was proposed by ~\cite{Markowitz1952}. In the
%Markowitz model, the ``return'' of a portfolio is given by the
%expected return of the portfolio and the ``risk'' of the portfolio is
%measured by the variance of the return of the portfolio. The variance
%is a good measure of risk only if the returns are {\em elliptically}
%distributed.  % \textcolor{red}{The returns on equity, at least for
%   short time horizons, can be approximated by a Normal random
%   variable}; consequently, the variance is an adequate measure for the
% risk in the portfolio.
%However, when the distribution of the returns
%of the underlying assets is not elliptic, variance is {\em not} an
%adequate risk measure.
Moreover, since variance is not sensitive to the
tails of the distribution, it is not a good measure of variability
when the returns are heavy tailed.

A number of  risk measures have been proposed in the
literature to accommodate asymmetry and also capture the effects of
heavier tails.  The Value-at-Risk $\text{VaR}_{\beta}(\tilde{L})$ at
the probability level $\beta$ for a random loss $\tilde{L}$ is defined
as the $\beta$ quantile of the loss distribution, i.e. the probability
of observing losses larger than $\text{VaR}_{\beta}(\tilde{L})$ is at
most $1-\beta$~\citep{Jorion2006}. VaR is extensively
used in risk management applications, and it is the mandated risk measure
in the Basel-II accords. However, it has a number of
shortcomings. First, VaR only depends on the probability of tail
losses and not their location in the tail. Second, VaR is not a convex
risk measure; consequently, portfolio selection with VaR constraints
often results in integer programs that are hard to solve. 

Conditional Value-at-Risk $\text{CVaR}_{\beta}(\tilde{L}) =
\mE[\tilde{L}\mid \tilde{L} \geq
\text{VaR}_{\beta}]$~\citep{RockafellarUryasev2000} and Expected
Shortfall $\text{ES}_{\beta} = \frac{1}{1-\beta}\int_\beta^1
\text{VaR}_p(\tilde{L}) dp$~\citep{Acerbi2002_ES} are closely related
risk functions that address the two shortcomings of VaR listed above. CVaR and
ES are both coherent risk
measures~\citep{ArtznerDelbaenEberHeath1999}, i.e. they are convex and
positively homogeneous. \cite{Acerbi2002_ES} showed that the ES of a
portfolio can be estimated from samples of the losses on the underlying
assets by solving a linear program~(LP), and that the estimate converges to the
ES of the portfolio with probability 1. \citet{RockafellarUryasev2000}
showed a similar result for CVaR assuming that the loss distribution
of the portfolio is continuous at the $\beta$ quantile.
\cite{Acerbi2002} extended ES to the spectral risk measure
$M_{\phi}(\tilde{L}) = \int_0^1 \text{VaR}_{p}(\tilde{L})\phi(p)dp$,
where $\phi(p)$ is a non-increasing probability distribution function.
The spectral risk measure
$M_{\phi}(\tilde{L})$ is coherent and, in fact,
$\text{ES}_{\beta}(\tilde{L}) = M_{\hat{\phi}}(\tilde{L})$ with
$\hat{\phi}(p) = \frac{1}{1-\beta} \ones_{\beta \leq p \leq 1}$.
\citet{Acerbi2002} also showed that the finite sample estimate
$M_{\phi}^N = \sum_{k=1}^N \phi(\frac{k}{N}) L_{(N-k)}$, where
$L_{(k)}$ denotes the $k$-th order statistic of $N$ independent
and identically distributed~(IID) samples of the random loss $\tilde{L}$,
converges to $M_{\phi}(\tilde{L})$ with probability $1$.

From % the results in 
\citet{Acerbi2002}, it follows that the portfolio
selection problem where the ``return'' is given by the expected
return of the portfolio and the ``risk'' is given by a spectral
risk measure of the portfolio can
be approximated by an LP. \cite{RockafellarUryasev2000} established
such an LP-based approximation result for the mean-CVaR portfolio
selection problem. 
\citet{AgarwalNaik2004} showed that the mean-CVaR
portfolio selection results in superior portfolios as compared to the
mean-variance approach when the risk of the  assets is nonlinear in the
underlying risk factors, e.g. when the asset is a derivative written on a
primary asset.
However, the resulting LP is very ill-conditioned, and 
solving such LP,
particularly when the scenario size is large, is very difficult in
practice~(see, e.g. \citep{AlexanderColemanLi2006}).
\citet{lim2009fragility} showed that the solution of the
mean-CVaR portfolio problem is often very sensitive to
estimation errors, i.e. small errors in the estimation of the mean and
the return in the scenarios can get amplified in the choice of the optimal
portfolio. 
% One way to address this
This sensitivity can be addressed by % introduce model robustness by
imposing spectral risk constraints with respect to several different
parameter values and also different risk models.
Constraints with respect to multiple risk models have become
especially important after the 2008 financial
crisis~(see, e.g. \citep{ceria2009novel}).
However, imposing multiple spectral risk constraints increases the size of
the LP by such an extent that 
state-of-the-art solvers are unable to solve most practical instances
of the portfolio selection problem.

Our contributions in this paper are as follows:
\begin{enumerate}[(a)]
\item We propose a new first-order gradient based algorithm
  \textsc{SpecRiskAllocate} to solve portfolio selection problems with
  multiple spectral risk constraints that is significantly faster than the
  naive LP-based approach. We exploit two key features of the portfolio
  selection problem to construct this 
  algorithm. The first is that the constraints in the LP 
  formulation~\eqref{eq:portfolio-lp} are
  very loosely coupled in that the samples from a particular risk model
  only play a role in the corresponding constraint. Thus, one can improve
  the run time of the algorithm by dualizing these constraints, provided 
  feasibility is maintained. We show in Theorem~(\ref{thm:penaltyrep}) that we are able
  to recover feasible portfolios for finite values of the dual
  variables. The second feature we exploit is that, since the LP is in
  fact a finite sample approximation to the stochastic optimization
  problem, in practice one is not attempting to solve it to
  very high accuracy~(e.g. $10^{-12}$ relative   error) but rather one is
  satisfied with moderate 
  accuracy~(e.g. $10^{-3}$ relative error). 
% The second feature we exploit is that, given that one uses
%   samples to approximate the spectral risk measures, solving the portfolio selection
%   problem to a very high degree of accuracy is not necessary.
  This
  allows us to % sacrifice accuracy by
  smooth the LP into a smooth convex optimization problem,
  resulting in significantly  faster convergence.

\item  \textsc{SpecRiskAllocate}
  computes the optimal portfolio by solving a sequence of small
  separable convex quadratic programs~(QPs). Thus, portfolio managers would be able
  to solve spectral risk constrained portfolio selection problems using
  existing tools for solving mean-variance problems. The number of
  variables in each of the convex QPs is equal to the number of
  assets and, therefore, these problems can be solved very efficiently.
  In some cases, the optimal solution of the mean-variance
  subproblem can be written in closed form or computed by a one
  dimensional search. 
  \textsc{SpecRiskAllocate} is also able to solve portfolio
  selection problems where the objective is to maximize a weighted sum of
  the  expected return and either a weighted
  combination or the maximum of a set of spectral risk measures.

\item The experimental results in Section~\ref{sec:results} clearly
  show that \textsc{SpecRiskAllocate} is able to efficiently solve very large
  spectral risk constrained portfolio selection problems.
  For most practical instances, \textsc{SpecRiskAllocate} is
  at least one order of magnitude faster than the state-of-the-art LP
  solvers.
  Moreover, we show that, in contrast to the LP-based method, \algo\ is \emph{not}
  ill-conditioned. This is a side-benefit of smoothing the problem.
  ``Smoothing'' approximates the LP polytope by a convex set
  without corners; thus, ensuring that the optimal solution is a
  continuous function of the problem and, therefore, not ill-conditioned.

\item A popular method for introducing robustness against model
  uncertainty is to 
  impose spectral risk constraints with respect to several
  risk models (see e.g. \citet{axioma032} and \citet{axioma040}). 
  In Section \ref{sec:results}, we  show that
  \textsc{SpecRiskAllocate} is able to
  solve a hedging portfolio
  selection problem with spectral risk constraints corresponding to
  multiple risk models in a
  computationally tractable manner.
  
\end{enumerate}
\textsc{SpecRiskAllocate} is based on the proximal gradient algorithm
\FISTA~proposed by \citet{BeckTeboulle2009} (see also
~\cite{Nesterov2005}). The algorithm we propose is similar to the
% algorithm
one proposed by \cite{iyengar2013fast} in that
both these algorithms use Nesterov smoothing
techniques~\citep{Nesterov2005}. However, there are a number of key
differences between the two methods. The algorithm in \cite{iyengar2013fast}
is only able to solve a mean-CVaR problem and can be extended to solve a
mean-weighted CVaR problem; however, it is not able to compute solutions
for portfolio selection problems with CVaR (or, more generally, spectral
risk) \emph{constraints}. \textsc{SpecRickAllocate} uses a different smoothing
technique that allows us to scale the algorithm to solve very large
portfolio selection problems without encountering any numerical
difficulty. \cite{iyengar2013fast} were unable to solve large problem
instances because the algorithm proposed therein quickly becomes numerically
unstable.

 % However, we would like to
 %  highlight some differences.  The algorithm proposed in
 %  \cite{iyengar2013fast} solves the minimum risk mean-CVaR problem with
 %  Nesterov's gradient descent method~\citep{Nesterov2005}.  The
 %  algorithm we propose is an addaption of \FISTA~ that solves the more
 %  difficult portfolio selection with multiple risk constraints
 %  problem, and can scale to much larger problem sizes.  If we could
 %  modify Iyengar and Ma's algorithm to consider multiple CVaR measures
 %  in the objective then both algorithms would be equivalent in theory,
 %  although not necessarily in practice.  }
% \textsc{SpecRiskAllocate} can solve portfolio selection problems
% with multiple risk constraints and can scale to much larger problem
% sizes.

The rest of this paper is organized as follows. In
Section~\ref{sec:model} we introduce the generalized spectral risk
measures and define the generalized spectral risk constrained portfolio selection problem. In
Section~\ref{sec:algo} we construct the \textsc{SpecRiskAllocate}
algorithm. In Section~\ref{sec:results} we discuss the results of our
numerical experiments. Finally, in Section~\ref{sec:conc} we conclude with some final remarks.

\section{Single period portfolio selection problem}
\label{sec:model}
Suppose there are $n$ assets in the market. Let $\tilde{\bL} =
\left(\tilde{L}_1, \ldots, \tilde{L}_n\right)^\top \in \reals^n$ 
denote the random rate of loss on the assets. Let $\bx \in \reals^n$
denote the portfolio of the investor, i.e., $\ones^{\top}\bx =
\sum_{i=1}^nx_i = 1$. The rate of loss $\tilde{L}_x$ of portfolio
$\bx$ is given by $\tilde{L}_x = \tilde{\bL}^\top\bx$. In this
paper, we want to identify portfolios that lie on the Pareto optimal
frontier with respect to the expected return $-\mE[\tilde{L}_x]$
and a set of generalized spectral risk measures~\citep{Acerbi2002}.

Except for some special cases --e.g. when the random loss vector $\tilde{\bL}$ 
is a
linear function of the distribution of elliptically distributed risk
factors $\tilde{\bZ}$-- the distribution of the random portfolio loss
$\tilde{L}_x$ is hard to characterize explicitly. This is definitely
the case if the portfolio $\bx$ contains derivative securities whose
distribution is nonlinear in the underlying risk factors. In practice,
$\tilde{\bL}$ is approximated by $N$ samples $\{\boldsymbol{\ell}_1,
\ldots, \boldsymbol{\ell}_N\}$ generated by some scenario
generator~(see, e.g. \cite{KoskosidisDuarte1997}).
Let $\bL = \left( \boldsymbol{\ell}_1, \ldots, \boldsymbol{\ell}_N \right)^\top \in \reals^{N \times n}$ denote the empirical loss
matrix, where the $j$-th column represents the vector of $N$ loss
realizations of asset~$j$. Thus, the random loss $\tilde{L}_x$ on the
portfolio $\bx$ can be approximated by the set of samples
$\{\boldsymbol{\ell}_1^\top\bx, \ldots, \boldsymbol{\ell}_N^\top\bx\}$
or, equivalently, by the vector $\bL\bx$. In the rest of this section, we
define the generalized spectral loss function for the vector $\bL\bx$
and relate it to the Expected Shortfall measure. This relation will be
important for designing our solution algorithm in Section~\ref{sec:algo}.

\subsection{Generalized spectral risk measures}
Let $\by = (y_1, \ldots, y_N)^\top$ denote $N$ samples of a random
variable $\tilde{Y}$. Let $\{y_{(\ell)}: \ell = 1, \ldots, N\}$
denote the order statistics of vector $\by$.
  
\begin{definition}[Expected shortfall (ES) \citep{Acerbi2002_ES}]
   The expected
  shortfall of $\by$ at level $\beta \in [0,1)$ is the average of the
  $\kappa = \lceil (1-\beta) N \rceil$ largest values of $\by$, i.e.,
  \eq
    \ES{\beta}{\by} = \frac{1}{\kappa} \sum_{\ell = N-\kappa
      +1}^N y_{(\ell)}.
  \en
\end{definition}

It is easy to check that $\ES{\beta}{\by}$
has the following variational characterization~(see,
e.g. \cite{ArtznerDelbaenEberHeath1999,RockafellarUryasevZabarankin2002,LuthiDoege2005})
: \eq \ES{\beta}{\by} =
\begin{array}[t]{rl}
  \max & \sum_{\ell=1}^N q_\ell y_\ell,\\
  \text{such that} &  \ones^\top \bq = 1,\\
  & \bo \leq \bq \leq \frac{1}{\kappa}\cdot \ones.
\end{array}
\en
Using linear programming duality~\citep{bertsimas1997introduction} it follows that
\begin{equation}
  \label{eq:ES-dual}
  \ES{\beta}{\by} =    \min_z \left\{ z + \frac{1}{\kappa} \cdot
    \sum_{\ell=1}^N (y_\ell-z)^+\right\},
\end{equation}
where $v^+ = \max\{v,0\}$.  \citet{Acerbi2002_ES} established that
$\ES{\beta}{\cdot}$ is a coherent risk
measure~\citep{ArtznerDelbaenEberHeath1999} and converges to CVaR~\citep{RockafellarUryasevZabarankin2002,LuthiDoege2005}
when the cumulative distribution function $F_Y(\cdot)$ of the random
variable $\tilde{Y}$ is continuous at $y = \inf\{x: F_Y(x)\geq
\beta\}$.

% The conditional value at risk $\CVaR{\beta}{\hat{\bL}\bx}$ at the
% probability $\beta \in [0,1]$ of the portfolio $\bx$ is defined as
% \eq \CVaR{\beta}{\hat{\bL}\bx} = \mE^{\mP}[\hat{\bL}\bx \mid
% \hat{\bL}\bx \geq \cF_{\hat{\bL}\bx}^{-1}(1-\beta)], \en where the
% random variable $\hat{\bL}\bx$ denotes the loss on the portfolio
% $\bx$, and has probability measure $\mP$ and cumulative density
% function~(CDF) $\cF_{\hat{\bL}\bw}$.  Thus, the CVaR is conditional
% expectation of the highest $(1-\beta)$-quantile of the random
% portfolio loss.

% CVaR has the following variational characterization
% ~\citep{ArtznerDelbaenEberHeath1999,RockafellarUryasevZabarankin2002,LuthiDoege2005}
% \eq \CVaR{\beta}{\hat{\bL}\bx} =
% \begin{array}[t]{rl}
%     	\max & \mE^\mQ(\hat{\bL}\bx) \\
%     	% 	    \text{s.t.} & 0 \leq \frac{\partial \mQ}{\partial
%     	%       \mP} \leq \frac{1}{\beta} \\
%     	%    & \mQ \mbox{ is a probability measure.}
% 	\end{array}
%  \en

%  Then \eq \CVaR{\beta}{\bL\bx} \approx
% 	\begin{array}[t]{rl}
%     	\max & \bq^{\top} \bL \bx \\
%     	% 	    \text{s.t.} & 0 \leq \bq \leq \frac{\bp}{\beta} \\
%     	% 	    & \ones^{\top} \bq = 1.
% 	\end{array}
%  \en

\begin{definition}[Spectral risk measure \citep{Acerbi2002}] 
  Let $\bfomega =
  (\omega_1, \ldots, \omega_N)^\top$ denote a non-decreasing
  probability mass function, i.e. $\bfomega \geq \bo$,
  $\ones^\top\bfomega = 1$, and $\omega_k \geq \omega_\ell$ whenever
  $k \geq \ell$.  The spectral risk measure $M_{\omega}(\by)$
  generated by $\bfomega$ is defined as
  \eq
    M_{\omega}(\by) = \sum_{\ell=1}^{N} \omega_{\ell}y_{(\ell)}.
    \label{spectlossfn}
  \en
  % where
  % \BIT
  % \item $X_{(l)}$ is the $l$-th order statistic (in particular,
  % 	 $X_{(1)} = \min\limits_{l=1,\cdots,N}X_{l}$ and $X_{(n)} =
  % 	 \max\limits_{l=1,\cdots,N}X_{l}$), and
  % \item $\omega_{l}\geq 0$ for $l=1,\cdots,N$,
  % 	 $\sum\limits_{l=1}^{N} \omega_{l} = 1$, and $\omega_{l}\geq
  % 	 \omega_{k}$ for $l\geq k$.  \EIT
\end{definition}
Let $\omega_{0} = 0$. Then,
\begin{align*}
  M_{\omega}(\by) & = \sum_{\ell=1}^{N} \omega_{\ell}y_{(\ell)}
  	= \sum_{\ell=1}^N (\omega_{\ell} - \omega_{\ell-1})
  \left(\sum_{j = \ell}^N 
    y_{(j)}\right)
  = \sum_{\ell=1}^N \gamma_{\ell} \ES{\beta_\ell}{\by},
\end{align*}
where $\gamma_{\ell} = (N-\ell+1)(\omega_{\ell}-\omega_{\ell-1})\geq
0$ and $\beta_{\ell} = \frac{\ell-1}{N}$. Hence,
it follows that $M_{\omega}(\by)$ is a coherent risk measure. It is
easy to check that $\sum_{\ell=1}^N \gamma_\ell = \sum_{\ell=1}^N
\omega_\ell = 1$, i.e. $\bfgamma$ is a probability mass function.
This motivates the following definition.
\begin{definition} [Generalized spectral risk measures]
 Let $\bfgamma \in
  \reals^d$ denote a probability mass function, i.e. $\bfgamma \geq
  \bo$ and $\ones^\top\bfgamma =1$. Let $\bfbeta \in [0,1)^d$.  The
  generalized spectral risk measure $\rho_{\gamma,\beta}(\by)$ is defined as 
  $$\rho_{\gamma,\beta}(\by) =
  \sum_{\ell=1}^{d} \gamma_{\ell} \ES{\beta_{\ell}}{\by}.$$
\end{definition}
%Note that the spectral risk measure~(\ref{spectlossfn}) is a special
%case of the generalized spectral risk measure~(\ref{eq:genspectralrisk}) with $d=N$,
%$\gamma_{\ell} = (\omega_{\ell}-\omega_{\ell-1})(N-\ell+1)$ and
%$\beta_{\ell} = \frac{\ell-1}{N}$.

\subsection{Portfolio selection problem}
We measure the risk of portfolio $\bx$ using $m$ different risk
models. Let $\bL_k \in \reals^{N_k\times n}$ denote the empirical loss
matrix corresponding to the $k$-th risk model, where $N_k$ denotes the
number of samples drawn according to the $k$-th model. The risk of
portfolio $\bx$ according to the $k$-th model is captured by a generalized
spectral risk measure $\rho_{\gamma_k,\beta_k}(\bL_k \bx)$, $k = 1, \ldots,
m$. In the remainder of this paper, we will abbreviate
$\rho_{\gamma_k,\beta_k}$ simply as $\rho_k$.

The goal of the spectral risk constrained portfolio selection problem is to find the portfolio $\bx$ that maximizes the expected return. Let $\bfmu \in \reals^n$ be the mean return vector. $\bfmu$ is typically set
equal to the weighted average $\bfmu = -\sum_{k=1}^m q_k \frac{1}{N_k} 
(\bL_k^\top\ones)$, where $\bq$ is a probability mass function that assigns 
weights to the $m$ risk models. Hence, the expected return of portfolio $\bx$ 
is $\bfmu^\top \bx$. Given that cardinality constraints are important in 
practice to control the transaction costs~\citep{chang2000heuristics}, we are 
interested in selecting sparse portfolios, i.e. portfolios whose $\ell_0$-norm 
$\sum_{i=1}^n \mathbbm{1}(\abs{x_i}>0)$
is small. Unfortunately, the associated
cardinality constrained portfolio selection problem is typically
NP-hard.  Nonetheless, a good approximation is to replace the $\ell_0$-norm with
the $\ell_{1}$-norm $\sum_{i=1}^n
\abs{x_i}$~\citep{candes2008enhancing}.
Thus, the spectral risk constrained sparse portfolio selection problem we want to solve is of the form:
% \begin{equation}
% \begin{array}{rl}
%     	\max & \bfmu^{\top}\bx \\
%     	% 	    \text{s.t.} & \rho_{\gamma_k,\beta_k}(\bL_{k}\bx)
%     	%     \leq \alpha_{k} \quad k=1,\cdots,m\\
%     	%    & \ones^{\top}\bx =1 \\
%     	%   & \norm{\bx}_\infty \leq B,\\
% \end{array}
% \end{equation}
% where $\bL_{k} \in \reals^{N_{k}\times n}$ is the empirical loss
% matrix used for constraint $k$ and $\mu$ is the expected return
% vector.  In practice, one is interested in sparse portfolios, but
% computing sparse portfolios is NP-hard.
\begin{equation}
  \begin{array}{rll}
    \max & \bfmu^{\top}\bx - \lambda \norm{\bx}_{1}\\
    \text{s.t.} & \rho_{k}(\bL_{k}\bx) \leq \alpha_{k}, & k=1,\cdots,m,\\
    & \ones^{\top}\bx =1, \\
    & \norm{\bx}_\infty \leq B,\\
  \end{array}
  \label{eq:spectralrisk}
\end{equation}
where $\lambda \geq  0$ is the parameter controlling the sparsity of the portfolio, $\alpha_k$ is the risk budget in the $k$-th risk model, the $\ell_{\infty}$-norm is defined as
$\norm{x}_{\infty} = \max_{1\leq i \leq n} \abs{x_i}$, and the bound
$B>0$ controls the leverage of the portfolio. There are two additional
interpretations for the  $\ell_1$-norm regularization in
\eqref{eq:spectralrisk}. Since $\ones^\top \bx = 1$, the 
$\ell_1$-norm $\norm{\bx}_1 = \sum_{i:x_i > 0} x_i - 
\sum_{i:x_i < 0} x_i = 1 - 2\sum_{i:x_i < 0} x_i$, and therefore,
penalizing the $\ell_1$-norm  is equivalent to penalizing short
positions~\citep{demiguel2009generalized}. Penalizing the $\ell_1$-norm of  
the portfolio also helps improve the out-of-sample 
performance of the portfolio in the presence of parameter estimation 
errors~\citep{demiguel2009generalized}.  In practice, the 
parameter $\lambda$ is 
chosen by cross-validation~\citep{demiguel2009generalized}
on the particular desired performance.
In this paper, we are agnostic to the portfolio 
manager's reasons for penalizing the $\ell_1$ norm of the portfolio 
--controlling transaction costs, constraining short sales, or improving 
out-of-sample performance of the portfolio. Therefore, we set $\lambda = \sfrac{2 \abs{\bfmu^\top
		\bx^\ast}}{\norm{\bx^\ast}}$ where 
$\bx^\ast \in \argmax\{\bfmu^\top 
\bx:\ones^{\top}\bx=1,\;\norm{\bx}_{\infty}\leq
B\}$ to ensure that the two terms in the objective are always comparable.
The numerical results reported in Section~\ref{sec:results}
clearly show that the running time of \algo\ is not dependent on the value of 
$\lambda$.

The solution method that we develop in Section~\ref{sec:algo}
is also able to solve the following portfolio selection problems:
\begin{enumerate}[(a)]
\item Sparse weighted mean-spectral risk portfolio selection problem
  \eq
    \begin{array}{rl}
      \max & \bfmu^{\top}\bx - \lambda \norm{\bx}_{1} -
      \sum\limits_{k=1}^{m}\theta_{k}\rho_{k}(\bL_{k}\bx) \\ 
      \text{s.t.} & \ones^{\top}\bx =1, \\
      & \norm{\bx}_\infty \leq B,\\
    \end{array}
    \label{eq:weightedproblem}
  \en
  where $\bftheta \in \reals^{m}_{+}$ is a vector of weights.
\item Sparse mean-max spectral risk portfolio selection problem
  \eq
    \begin{array}{rl}
      \max & \bfmu^{\top}\bx - \lambda \norm{\bx}_{1} - \theta\left(
        \max\limits_{k=1,\cdots,m}\rho_{k}(\bL_{k}\bx) \right)\\ 
      \text{s.t.} & \ones^{\top}\bx =1,\\
      & \norm{\bx}_\infty \leq B,\\
    \end{array}
    \label{eq:maxproblem}
  \en
  where $\theta\geq 0$ is a penalty on the maximum spectral risk measure.
\end{enumerate}

From the dual representation~\eqref{eq:ES-dual} of ES, it follows that the portfolio selection
problem~\eqref{eq:spectralrisk} can be reformulated as
\eq
  \begin{array}{rll}
    \max & \bfmu^{\top}\bx - \lambda \norm{\bx}_{1}\\
    \text{s.t.} & \sum\limits_{\ell = 1}^{d_k} \gamma_{k\ell} \left(z_{k\ell} +
      \frac{1}{(1-\beta_{k\ell})N_k} \sum\limits_{j = 1}^{N_k}\left((\bL_k \bx)_{j}
        - z_{k\ell}\right)^+ 
    \right) \leq \alpha_{k}, & k=1,\cdots,m,\\ 
    & \ones^{\top}\bx =1,\\
    & \norm{\bx}_\infty \leq B,\\
  \end{array}
\en
where $(\bL_k\bx)_{j}$ denotes the $j$-th component of the vector
$\bL_k\bx \in \reals^{N_k}$.
% , and we use the fact that $\rho_{k}(\bL_k\bx) = \sum_{\ell =
%   1}^{d_k}\gamma_{k\ell} \ES{\beta_{k\ell}}{\bL_k\bx}$.
By introducing new variables $y_{jk\ell} = \left((\bL_k \bx)_j -
  z_{k\ell}\right)^+$, and $\xi_i = \abs{x_i}$, the above optimization
problem can be reformulated as the
LP
\begin{equation}
\begin{array}{rll}
  \max & \bfmu^{\top}\bx - \lambda \ones^\top\bfxi\\
  \text{s.t.} & \sum\limits_{\ell = 1}^{d_k} \gamma_{k\ell} \left(z_{k\ell} +
    \frac{1}{(1-\beta_{k\ell})N_k} \sum\limits_{j = 1}^{N_k}  y_{jk\ell} 
  \right) \leq \alpha_{k}, & k=1,\cdots,m,\\ 
  & y_{jk\ell} \geq (\bL_kx)_j - z_{k\ell}, \quad \quad \;
  j = 1, \ldots, N_k, & \ell = 1, \ldots, d_k, \quad k = 1, \ldots,m,\\
  & \bfxi \geq \bx, \quad \quad \; \bfxi \geq - \bx,\\
  & \ones^{\top}\bx =1, \quad \norm{\bx}_\infty \leq B,\\
  & \by \geq 0.
\end{array}
  \label{eq:portfolio-lp}
\end{equation}

Unfortunately, this LP is typically very large. For example, when each
generalized risk measure $\rho_k$ has $d$ ES components,
and the number of samples $N_k$ is equal to $N$ for each $k$, the
LP~\eqref{eq:portfolio-lp} has $\cO(mdN + n)$ variables and
constraints. Thus, with $n = 100$ assets, $m = 5$ risk constraints, each
with $d = 3$ ES components, and $N = 10,000$ samples, the LP has
$150,100$ variables even though the original portfolio selection
problem has only $n = 100$ variables! In addition, at any optimal
solution a very large fraction of the $y_{jkl}$ variables are zero;
consequently, the LP is very ill-conditioned. Large, ill-conditioned
LPs are extremely hard to solve in practice. In
Section~\ref{sec:results} we give empirical evidence supporting this
claim.

\section{Spectral risk constrained portfolio selection algorithm}\label{sec:algo}
  
In this section, we propose a fast iterative algorithm \algo~for
computing a solution to \eqref{eq:spectralrisk}
% that directly solves a penalty formulation for
% \eqref{eq:spectralrisk}
without introducing any new variables. Our goal is to be able to scale
\algo~ to solve very large scale portfolio selection problems;
therefore, we restrict ourselves to gradient descent
algorithms. \algo~ is an application of the proximal gradient
algorithm \FISTA~\citep{BeckTeboulle2009} to a suitably defined
``smoothed'' penalty reformulation of \eqref{eq:spectralrisk}.
% Our main methodological contribution in this paper is
In
  Theorem~\ref{thm:penaltyrep} % where
  we establish an explicit value
  for the penalty parameter that guarantees that an $\varepsilon$-optimal
  solution to \eqref{eq:spectralrisk} can be reconstructed from the
  solution to the penalty formulation.
% We need to solve separable convex QPs to compute gradients of the
% smoothed penalty functions and the \FISTA~ iterates.
The numerical results in Section~\ref{sec:results} clearly show that
our algorithm, which solves several small convex QPs, is significantly
faster than the LP formulation that solves one very large LP. \algo\ can be 
viewed as a decomposition algorithm that decomposes
the large LP into a number of small QPs by exploiting the fact that 
its constraints are very loosely coupled, and then smooths the smaller
QPs to improve convergence.

%In the rest of this section, we construct the specific penalty
%formulation that we use. We show how to smooth this penalty term to
%ensure that it has a Lipschitz continuous derivative. Next, we
%describe how to use \FISTA~ to iteratively solve this smoothed penalty
%problem. We also establish an upper bound on the penalty parameter
%that guarantees an $\varepsilon$-optimal solution.

\subsection{Smoothed penalty formulation}
% In this section, we describe a fast iterative algorithm for
% computing a solution of \eqref{eq:spectralrisk}. At each step, this
% algorithm computes the solution of the $\ell_1$-norm constrained
% mean-variance problem of the form \eq
% \begin{array}{rl}
%   \max & \varphi^{\top}\bx - \frac{\rho}{2}\bx^{\top}\bx - \kappa \norm{\bx}_{1}\\
%   \text{s.t.} & \ones^\top \bx = 1,\\
%   & \norm{\bx}_{\infty} \leq B.
% \end{array}
% \en We show that this problem can be solved by a simple search.  In
% practice, the portfolio selection problem is likely to have other
% linear constraints $\bC\bx \leq \bc$. In that case, the iterative
% algorithm has to solve a $\ell_1$-norm constrained mean-variance
% problem \eq
% \begin{array}{rl}
%   \max & \varphi^{\top}\bx - \frac{\rho}{2}\bx^\top\bx - \kappa \norm{\bx}_{1}\\
%   \text{s.t.} & \ones^\top \bx = 1,\\
%   & \bC \bx \leq \bd,\\
%   & \norm{\bx}_{\infty} \leq B.
% \end{array}
% \en Note that the dimension of the problem is determined by the size
% of the portfolio $\bx$ -- the number of risk models, and the samples
% in each risk model do not affect the size of the problem.  Suppose
% the cardinality of the portfolio $\bx$ is not of concern. Then this
% problem reduces to a very simple mean-variance portfolio selection
% where the assets are completely uncorrelated. Thus, it follows that
% if the portfolio manager has access to a mean-variance solver, she
% can use our method to compute mean-spectral risk optimal portfolios.
The portfolio selection problem~\eqref{eq:spectralrisk} is clearly
equivalent to the problem \eq
\begin{array}{rll}
  \max & \bfmu^{\top}\bx - \lambda \norm{\bx}_{1}\\
  \text{s.t.} & \max_{1 \leq k \leq m}\left\{\rho_{k}(\bL_{k}\bx) -
    \alpha_{k}\right\} \leq 0, & k=1,\cdots,m,\\
  & \ones^{\top}\bx =1, \\
  & \norm{\bx}_\infty \leq B.\\
\end{array}
\en An exact penalty formulation of this optimization problem is given
by \eq
\begin{array}[t]{rl}
  \min & \eta\left(\lambda \norm{\bx}_{1}  - \bfmu^{\top}\bx\right) +  \left(\max_{1
      \leq k \leq m}\left\{\rho_{k}(\bL_{k}\bx) - \alpha_{k}\right\}\right)^+\\
  \text{s.t.} &  \ones^{\top}\bx =1, \\
  & \norm{\bx}_\infty \leq B,
\end{array}
\en where $\eta$ denotes the penalty parameter. We will find it
convenient to scale the objective by $\eta$ instead of scaling the
penalty term. Let us express the maximum of $m+1$ values, $t_1, \ldots, t_{m+1}$, as $\Psi(t_{1},\cdots,t_{m+1})  = \max_{\bu}\left\{\bt^{\top}\bu: \ones^{\top}\bu = 1, \bu \geq \mathbf{0}\right\}$, and define $g(\bx) = \Psi(\rho_{1}(\bL_{1}\bx)-\alpha_{1}, \ldots,
  \rho_{m}(\bL_{m}\bx)-\alpha_{m},0)$.
Then, the above exact penalty formulation can be written as
\begin{equation}
  \label{eq:penalty_opt}
  G(\eta) =   \begin{array}[t]{rl}
    \min & \eta\left(\lambda \norm{\bx}_{1}  - \bfmu^{\top}\bx\right) + g(\bx)\\  
    \text{s.t.} & \ones^{\top}\bx =1, \\
    & \norm{\bx}_\infty \leq B.\\
  \end{array}
\end{equation}

We expect that the solution to~(\ref{eq:penalty_opt}) will converge
to a solution to~(\ref{eq:spectralrisk}) as $\eta
\rightarrow 0$. The next result establishes
this claim and shows that there exists a lower bound $\eta^\ast$ for the penalty parameter that
guarantees that one can construct an $\varepsilon$-optimal solution for
(\ref{eq:spectralrisk}) from an $\varepsilon$-optimal solution to an
appropriately smoothed version of $G(\eta^\ast)$.
% Next, we show that we do not have to set $\eta$ arbitrarily small in
% order to ensure $\varepsilon$-optimality. In fact, we can guarantee
% $\varepsilon$-optimality for a finite $\eta$.
\begin{theorem}[Penalty Representation]
  \label{thm:penaltyrep}
  Suppose there exists a portfolio $\bz$, $\ones^\top \bz = 1$,
  $\norm{\bz}_{\infty} \leq B$, such that $\bz$ strictly satisfies all
  the generalized spectral risk constraints, i.e. $\rho_k(\bL_k \bz) <
  \alpha_k$, for $k = 1, \ldots, m$. Define $g_{\max}(\bx) = \max_{1 \leq k \leq
      m}\{ \rho_k(\bL_k\bx)- \alpha_k\}$.
  Let $P_u$ denote any upper bound on the optimal value $P^\ast$ of
  the spectral risk portfolio selection
  problem~(\ref{eq:spectralrisk}). Suppose $\overline{\bx}$ is
  an $\varepsilon$-optimal solution to the penalized
  problem~(\ref{eq:penalty_opt}) with 
  \eq \eta^\ast =
  \frac{\abs{g_{\max}(\bz)}}{P_u - (\bfmu^\top\bz - \lambda
    \norm{\bz}_1)}. \en
  Then, 
  \eq \hat{\bx} =
  \frac{1}{1+\theta} \cdot \overline{\bx} + \frac{\theta}{1+\theta}
  \cdot \bz \en
  is an $\varepsilon$-optimal solution to the spectral risk
  portfolio selection problem~\eqref{eq:spectralrisk}, where $\theta =
      \max\left\{\sfrac{g_{\max}(\overline{\bx})}{\abs{g_{\max}(\bz)}}, 0
      \right\}$.
\end{theorem}
\begin{proof}
  The proof is identical to that of Theorem~2 in
  \cite{IyengarPS11:packing}.
\end{proof}

% \noindent In practice, we stop whenever
% $\frac{\norm{\bx^{(j)}-\bx^{(j-1)}}_{2}}{\norm{\bx^{(j-1)}_{2}}} <
% \varepsilon$, and the iterate $\bx^{(j)}$ is $\varepsilon$-feasible,
% i.e. $g_{\max}(\bx) \leq \epsilon$.
%Finally, note that for fixed $\eta$,~(\ref{eq:penalty_opt}) is
%equivalent to ~(\ref{eq:maxproblem}). Thus, it is immediately obvious
%that our method will also solve~(\ref{eq:maxproblem}).

We would like to use a gradient-based algorithm to solve
problem~(\ref{eq:penalty_opt}). However, both $\Psi$ and the spectral
risk measure $\rho$ are non-smooth functions of their argument;
consequently, $g(\bx) = \Psi(\rho_{1}(\bL_{1}\bx)-\alpha_{1}, \ldots ,
\rho_{m}(\bL_{m}\bx)-\alpha_{m},0)$ is a non-smooth function of the
portfolio $\bx$.  We use a smooth approximation
$g_{\nu\delta}(\bx)$ to the function $g(\bx)$ such that $g(\bx) -
\nu - \delta \leq g_{\nu\delta}(\bx) \leq g(\bx)$. The details
of the construction of $g_{\nu\delta}$ are given in Appendix
\ref{app:smoothing}.  By replacing $g(\bx)$
in~\eqref{eq:penalty_opt} with $g_{\nu\delta}(\bx)$, we obtain the
following smooth optimization problem:
%\begin{equation}
\eq
  G_{\nu\delta}(\eta) = \begin{array}[t]{rl}
    \min & \eta\left(\lambda \norm{\bx}_{1} - \bfmu^{\top}\bx\right) +
    g_{\nu\delta}(\bx)  \\ 
    \text{s.t.} & \ones^{\top}\bx =1, \\
    & \norm{\bx}_\infty \leq B.\\
  \end{array}
%  \label{eq:penalty_opt_smooth}
%\end{equation}
\en
Since the scenario-based spectral risk portfolio selection problem is
itself an approximation to the stochastic optimization problem where
the distribution of the loss $\tilde{\bL}$ is known, one does not expect to
solve these problems to very high accuracy, i.e. a solution error of the
order of $10^{-12}$. In practice, error of the order of $10^{-3}$
is sufficient. 
%  solving
% the spectral risk portfolio selection
% problem~(\ref{eq:spectralrisk}) to very high accuracy is not of
%   practical importance, i.e. replacing the $g(\bx)$ by
% $g^{\nu\delta}(\bx)$, for appropriately chosen small
%   values for $\nu$ and $\delta$, is unlikely to make a difference in
% practice.
Therefore, solving the smoothed problem for appropriately
chosen values of $\nu$ and $\delta$ is sufficient for most practical
instances. Moreover, in Section~\ref{sec:results} we show that the smoothing significantly improves the
computational tractability of this problem.

\subsection{First-order proximal gradient algorithm}
\label{sec:conv}
\algo~ is displayed in Algorithm~\ref{fig:algorithm}.
\algo~computes an $\varepsilon$-optimal solution
for the spectral risk constrained portfolio selection problem
\eqref{eq:spectralrisk} by approximately solving a sequence of
smoothed penalty problems $G_{\nu\delta}(\eta)$ for a decreasing
sequence of $\eta$. We begin with $\eta \gets \eta_0$ and then progressively
reduce $\eta \gets c_{\eta}\eta$, where $c_{\eta}<1$. This
continuation scheme ensures that \algo~ is able to take large steps
when the iterates are far from optimality.  In
Theorem~\ref{thm:penaltyrep} we showed that there exists $\eta^\ast >
0$ such that we can recover an $\varepsilon$-optimal solution for
(\ref{eq:spectralrisk}) by solving $G_{\nu\delta}(\eta^\ast)$, i.e. we
do not have to drive $\eta$ all the way to zero.  This feature adds
stability to \algo~ since the numerical accuracy required to solve $G_{\nu \delta} (\eta)$ increases as
$\eta \searrow 0$ (see e.g. \citet{nocedal1999numerical}).
\noindent In practice, we stop whenever the relative change in iterate $\bx^{(j)}$ is
smaller than the tolerance $\varsigma$,
and the iterate $\bx^{(j)}$ is $\varsigma$-feasible,
i.e. $g_{\max}(\bx^{(j)}) \leq \varsigma$.  \algo~ calls \FISTA~to
approximately solve $G_{\nu\delta}(\eta)$ for a fixed value of
$\eta$. \FISTA~is a proximal gradient method, i.e. a gradient
descent algorithm with an additional proximal term to control the step
length. The parameter $\tau$ controls the accuracy demanded by
\FISTA. We need $\tau \searrow 0$ to ensure that the accuracy is
increased as $\eta \searrow 0$.

% however, we terminate the algorithm early using a relative error
% criterion -- see Line~\ref{line:terminate} in
% Figure~\ref{fig:algorithm}.

\begin{algorithm}[tp]
  \begin{algorithmic}[1]
    \STATE $\eta \gets \eta_0$
    \STATE $\tau \gets \tau_0$
    \STATE $C \gets 1$
    \STATE $\bx \gets \frac{1}{n}\ones$
    \REPEAT
    	\STATE $\hat{\bx} \gets \bx$
    	\STATE $(\bx, C) \gets \FISTA(\hat{\bx}, C, \eta, \tau, \nu, \delta)$
    	\STATE $\eta \gets c_{\eta}\eta$
    	\STATE $\tau \gets c_{\tau}\tau$
    \UNTIL{$\left(\norm{\bx-\hat{\bx}}_{2}/\norm{\hat{\bx}}_2 <
        \varsigma \right)$ \AND $\max_{1\leq k \leq m}\{
      \rho_k(\bL_{k}\bx)-\alpha_{k}\} <
      \varsigma$} \label{line:terminate}
    \RETURN $\bx$
  \end{algorithmic}
  \caption{Algorithm \textsc{SpecRiskAllocate}($\eta_0,c_{\eta},\tau_0,c_{\tau},\nu,\delta,\varsigma$)}
  \label{fig:algorithm}
\end{algorithm}

\begin{algorithm}[tp]
  \begin{algorithmic}[1]
	\STATE $\zeta \gets 1.5$
    \STATE $t \gets 1$
    \STATE $\by \gets \bx$
    \REPEAT
    	\STATE $\hat{\bx} \gets \bx$
    	\STATE $\hat{t} \gets t$
    	\STATE $\bfxi \gets \textsc{ComputeGradient}(\by, \nu, \delta)$
    	\REPEAT \label{line:begin-backtracking}
   	 		\STATE $\bx \gets \argmin \left\{ \eta\lambda\norm{\bz}_{1} +
      \bfxi^\top(\bz-\by) + \frac{C}{2}\norm{\bz-\by}_2^2:
      \ones^{\top}\bz=1, \; \norm{\bz}_{\infty}\leq B \right\}$
      		\STATE $F \gets - \eta\bfmu^{\top}\bx + \eta\lambda\norm{\bx}_{1} +
    g_{\nu\delta}(\bx)$
    		\STATE $Q \gets \eta\lambda\norm{\bx}_{1} -
    \eta\bfmu^{\top}\by + g_{\nu\delta}(\by) + \bfxi^\top(\bx-\by) +
    \frac{C}{2}\norm{\bx-\by}_2^2 $
		    \STATE $C \gets C \zeta$
    	\UNTIL{$F < Q$} \label{line:end-backtracking}
    	\STATE $C \gets C / \zeta$
    	\STATE $t \gets \frac{1+\sqrt{1+4\hat{t}^{2}}}{2}$
    	\STATE $\by \gets \bx + \frac{\hat{t}-1}{t}(\bx-\hat{\bx})$
   	\UNTIL $\Big(\norm{\bx-\hat{\bx}}_2/\norm{\hat{\bx}}_2\Big) \leq \tau$
    \RETURN $(\bx,C)$
  \end{algorithmic}
  \caption{Function \FISTA($\bx$, $C$, $\eta$, $\tau$, $\nu$, $\delta$)}
  \label{fig:FISTA}
\end{algorithm}

Next we describe some of the essential features of \FISTA.  We refer
the reader to
\cite{BeckTeboulle2009} for the details of the algorithm.  The
particular implementation of \FISTA~that we employ is displayed in
Algorithm~\ref{fig:FISTA}.
% in Appendix~\ref{app:specriskalgo}.
\FISTA~computes an approximate solution to $G_{\nu\delta}(\eta)$ by
iteratively solving a sequence of quadratic optimization problems of
the form
%[6]\begin{equation}
%  \label{eq:FISTA-iter-opt}
%  \begin{array}{rl}
%    \mbox{min} & \eta\lambda\norm{\bx}_1 + h(\bx, \by)\\
%    \mbox{s.t.} & \ones^\top \bx = 1,\\
%    & \norm{\bx}_{\infty} \leq B,
%  \end{array}
%\end{equation}
%where
%\begin{equation}
%  \label{eq:g-quad}
%  h(\bx, \by)  = \big(\eta
%    \bfmu + \nabla 
%      g_{\nu\delta}(\by) \big)^\top(\bx-\by)  + 
%  \frac{C}{2} \|\bx-\by\|_2^2,
%\end{equation}
%and $C$ is a bound on the Lipschitz constant of the gradient $\nabla
%g_{\nu\delta}(\by)$ [6].
\begin{equation}
	\label{eq:FISTA-iter-app}
	\begin{array}{rl}
	\mbox{min} & \eta\lambda \norm{\bx}_1 + \bfxi^\top(\bx-\by) +
	\frac{C}{2} \norm{\bx-\by}_2^2,\\
	\mbox{s.t.} & \ones^\top \bx = 1,\\
	& \norm{\bx}_{\infty} \leq B,
	\end{array}
\end{equation}
where $\bfxi = \nabla \left(-\eta \bfmu^\top \by +
g_{\nu\delta}(\by)\right) = -\eta \bfmu + \nabla
g_{\nu\delta}(\by)$, and $C$
is the Lipschitz constant of the gradient $\bfxi$.
Although one can explicitly compute its value,
it is often the case that the Lipschitz constant $C$ is too large.
In practice, it is more efficient to use a backtracking method to compute $C$.
The function \FISTA~does
backtracking in
lines~\ref{line:begin-backtracking}--\ref{line:end-backtracking} of
Algorithm~\ref{fig:FISTA}.
\FISTA~is guaranteed to converge to an $\varepsilon$-optimal solution in
$\cO(\sfrac{1}{\varepsilon})$ iterations. However, the worst-case bound is
often too conservative in practice. We terminate the \FISTA~iterations
whenever the relative change in the iterates is below a threshold
$\tau$. We make $\tau$ progressively tighter as $\eta$ is decreased.

Let $\by^{(k)}$ denote the current \FISTA~iterate. Since
$-\eta\bfmu^\top\bx + g_{\nu\delta}(\bx)$ is a convex function with a
Lipschitz continuous derivative, it follows that the quadratic
function $\bfxi^\top(\bx-\by)  + 
  \frac{C}{2} \|\bx-\by\|_2^2$
is an upper bound for $-\eta\bfmu^\top\bx + g_{\nu\delta}(\bx)$.
% From~\eqref{eq:grad-g} it follows that we need to solve several
% separable convex QPs in order to compute $\nabla g_{\nu\delta}(\by)$.
%
% In iteration $k+1$, \FISTA~ solves the optimization problem It is
% important to note that $h(\bx;\by^{(k)})$ is an {\em upper} bound on
% $g_{\nu\delta}(\bx)$.
This ensures that the improvement in the true objective at the new
iterate $\by^{(k+1)}$ is at least as large as that predicted by the
quadratic approximation~\eqref{eq:FISTA-iter-app}. The quadratic
approximation~\eqref{eq:FISTA-iter-app} 
%is a trade-off between having some
%second-order information and full Hessian information. (\ref{eq:g-quad})
only uses the first-order gradient information. 
Therefore, the algorithm used to solve $G_{\nu\delta}(\eta)$ can be scaled 
to much larger problem sizes, and is also
considerably more stable as the problem size increases, as compared to a 
full-fledged quadratic approximation that uses all the Hessian information; 
however, at
the cost of a larger iteration count. 
%[6] In practice, the Lipschitz
%constant $C$ is not known, and one has to use backtracking
%to estimate $C$. See Appendix~\ref{app:specriskalgo} for details.
%
Finally, note that \eqref{eq:FISTA-iter-app} is equivalent to
%[6]\begin{equation}
%  \label{eq:FISTA-mv}
%  \begin{array}{rl}
%    \mbox{min} & \eta\lambda \norm{\bx}_1 + \left(\eta \bfmu + \nabla
%      g_{\nu\delta}(\by) - C \by \right)^\top \bx +
%    \frac{C}{2} \bx^\top\bx\\
%    \mbox{s.t.} & \ones^\top \bx = 1,\\
%    & \norm{\bx}_{\infty} \leq B,
%  \end{array}
%\end{equation}[6]
\eq
\label{eq:FISTA-mv}
\begin{array}{rl}
	\mbox{min} & \eta\lambda \norm{\bx}_1 + (\bfxi - C \by)^\top \bx +
	\frac{C}{2} \bx^\top\bx,\\
	\mbox{s.t.} & \ones^\top \bx = 1,\\
	& \norm{\bx}_{\infty} \leq B,
\end{array}
\en	
i.e. the \FISTA~iterates are computed by solving an $\ell_1$-penalized
separable convex QP with the number of decision variables equal to the
number of assets.  Thus, this problem can be solved very efficiently
if one has access to a mean-variance solver.  In
Appendix~\ref{app:specriskalgo} we show how to solve this problem
using a single one-dimensional search.
In practical instances, where it is likely that the portfolio selection problem 
has additional linear constraints, the portfolio manager can use the 
mean-variance or quadratic solver to compute the \FISTA~iterates.
In Appendix~\ref{app:specriskalgo}, we also show how to compute the gradient 
$\bfxi$ using $\sum_{k=1}^m d_k+1$ one-dimensional searches.

\section{Numerical results}\label{sec:results}

In this section we present numerical experiments that show the advantage of
\algo~over the LP formulation when dealing with large instances of the
spectral risk constrained portfolio selection problem.
Next, we illustrate the convenience of considering several risk models to 
overcome the uncertainty in risk parameters when selecting a portfolio 
to hedge the risk of an existing one.

\subsection{Ill-Conditioning and Problem Scaling Results}
\label{sec:scale}
We tested our algorithm on random instances of the spectral risk
constrained portfolio selection problem~(\ref{eq:spectralrisk}). We generated
instances with different values for the number of assets~$n$. The
number of spectral risk constraints was $m=5$ for all instances. For
each spectral risk measure, we fixed the number of ES components to $d=3$.
The number of loss scenarios $N$ was set equal for all risk
models. We randomly generated the expected return percentage
vector $\bfmu$, the scenario-based loss
matrices $\bL_{k}$, the ES weight vectors $\bfgamma_{k}$, and the ES levels
$\bfbeta_{k} \in [0.9,1)^{d}$. The spectral risk budgets
$\alpha_{k}$ were set to $\hat{\alpha}_k -
0.1\abs{\hat{\alpha}_k}$, where $\hat{\alpha}_k$ is the value of the
$k$-th spectral risk measure $\rho_k(\bL_k \hat{\bx})$ at portfolio 
$\hat{\bx} = \sfrac{1}{n} \ones$. We set the leverage bound to $B = 1$, and the 
parameter 
controlling the sparsity of the portfolio either to $\lambda=0$ or $\lambda = 
\lambda^\ast$, where $\lambda^\ast = 
\sfrac{2\abs{\bfmu^{\top}\bx^\ast}}{\norm{\bx^\ast}_{1}}$, and $\bx^\ast = 
\argmax\{\bfmu^\top \bx:\ones^{\top}\bx=1,\;\norm{\bx}_{\infty}\leq
B\}$. For all the instances generated, the value of $\lambda^\ast$ was in
the interval $[0.01,0.03]$.  The \algo~parameters were set as follows
\eq
	\eta_0 = 10, \; c_{\eta} = 0.99, \; \tau_0 = 10^{-4}, \;
c_{\tau} = 0.95, \; \nu = 0.01\min|\alpha_{k}|, \; \delta =
0.01, \; \varsigma = 10^{-2}.  \en
We solved each instance of the spectral risk constrained sparse portfolio selection problem
using a MATLAB implementation of \algo. For each instance, we also
solved the LP formulation~\eqref{eq:portfolio-lp} using the
state-of-the-art LP solver~Gurobi~\citep{gurobi} with an optimality
tolerance of $\varsigma = 10^{-2}$. 
We solved the instances using  Gurobi version 5.0.2 and Gurobi
  version 5.6.0. Our results indicate that, although the performance
of Gurobi has improved significantly from one version to the other, our algorithm
still offers a significant advantage over this state-of-the-art LP solver.
We called Gurobi from MATLAB
using Gurobi's MATLAB interface.  MATLAB was run on a 6-core, 3.07GHz
Intel Xeon processor with 66GB of RAM running the Ubuntu OS.

As mentioned in Section \ref{sec:algo}, the LP
formulation~\eqref{eq:portfolio-lp} is very ill-conditioned. This is
manifested in a high variance in the number of iterations required to
solve similar problems, i.e. with very small perturbations in the
parameter values. 
%  Ill-conditioning of
% an LP can be reflected in a high variance of the number of iterations
% needed to compute an optimal solution to similar (perturbed) problems.
We now show empirically that one does not face this issue when
\eqref{eq:spectralrisk} is solved using~\algo. % ~to solve the smooth
% approximation of equivalent problem \ref{eq:spectralrisk}, we do not
% face this issue. 
We generated a base instance with
$(n,N)=(100,1000)$. Next, we created $S = 100$ perturbed instances by
setting
 % for $k=1,\cdots, m$, perturb 
each entry $\ell_{ijk}^s$ of the loss matrix $\bL_k^s$, corresponding the
to the $s$-th perturbed problem, to  
$\ell_{ijk}^s = \ell_{ijk} + t\abs{\ell_{ijk}}\varepsilon_{ijk}^s$, where $t \in
\{0.05, 0.1\}$ % by adding
% to it, where 
and $\varepsilon_{ijk}^s$ are I.I.D. standard Normal random variables.
% IID standard Normal random variables, and $t \in \{5,
% 10\}\%$. 
% Maintaining all other inputs unchanged, we generated 100
% perturbed problems for each $t$ and solved them using Gurobi and
% \algo. 
Table~\ref{tab:illcond} shows the mean $\mu_S$ and the standard deviation
$\sigma_S$ of the number of iterations required by Gurobi and by
\algo~(total \FISTA~iterations, in this case) to solve the $S=100$
perturbed instances. Table~\ref{tab:illcond} also shows the coefficient of
variation $\sfrac{\sigma_S}{\mu_S}$ of the number of iterations needed to
solve the perturbed instances. The number of iterations required by \algo\
has a coefficient of variation of less 
than $1$\%, where the same number for Gurobi 5.0.2 (resp. Gurobi 5.6.0) is
approximately $158$\% (resp. $11$\%).  It is clear that  the
ill-conditioning is completely resolved by \algo.

\begin{table}[tbp]
  \centering
  \footnotesize
  \begin{tabular}[h]{|c|c|r@{.}lr@{.}lr@{.}l|}
    \hline
    perturbation $t$ & solver & \multicolumn{2}{c}{$\mu_S$} &
    \multicolumn{2}{c}{$\sigma_S$} & \multicolumn{2}{c|}{$\sigma_S/\mu_S$}
    \\  
    \hline    	
    0.05 & Gurobi 5.0.2 & 132&53 & 202&75 & 1&5298 \\
    0.05 & Gurobi 5.6.0 & 111&77 &    4&30 & 0&0385 \\
    0.05 & \algo        &  82&12 &    0&41 & 0&0050 \\
    \hline	
    0.10 & Gurobi 5.0.2 & 107&14 & 169&25 & 1&5797 \\
    0.10 & Gurobi 5.6.0 & 118&65 &  13&59 & 0&1145 \\
    0.10 & \algo        &  81&96 &   0&75 & 0&0092 \\
    \hline
  \end{tabular}
  \caption{Mean, standard deviation, and coefficient of variation of the number of iterations needed to solve 100 perturbed problems. Variance is much higher in the Gurobi case than in the \algo~case due to ill-conditioning of the LP formulation.} 
  \label{tab:illcond}
\end{table}

In Section~\ref{sec:algo}, we argued that the number of constraints and
variables in LP~\eqref{eq:portfolio-lp} is very large. Consequently,
we expect the time to solve large instances 
% huge in the number of variables and constraints. Hence, instances with
% a large, though not very large, number of scenarios 
using the LP formulation to be high. In contrast, we expect \algo~to be able to solve large
instances in a very reasonable amount of time. To support these claims,
we generated $10$ random instances for each pair of parameters $(n,N)$ 
and solved them with the sparsity parameter $\lambda$ set equal to 
$\lambda^\ast$ or $0$.
Table~\ref{tab:results} reports the results for this problem scaling study.
The column labeled ``err'' lists the mean relative error of the optimal
value found by \algo~with respect to the one found by Gurobi. 
For all but the smallest-sized problem, i.e. $(n,N) = (10,100)$, \algo\ found
a solution with an objective value within $0.5\%$ of the optimal 
value, and an optimal solution for $7$ out of the $11$ problems parameterized
by $(n,N)$. % Even in the only exception,
% the objective value of the solution found by \algo~is, in average, only $3.1\%$ suboptimal.
For each instance, we set a maximum solution time limit of $1$ hour. 
The columns labeled ``limit''
list the number of instances that could not be solved within the time limit. The columns labeled
``time(s)'' list the average run time in seconds, where we have included a
run time of $3600$ seconds for those instances that reached the solution
time limit. Note that for three of the largest-sized problems, namely
$(n,N) \in \{(100,15000), (1000,10000),(1000,15000)\}$, Gurobi was unable to 
solve at
least $1$ instance and up to $9$ out 
of $10$ instances within the time limit. Although the running time of  Gurobi 5.6.0 shows a
remarkable 
improvement for smaller problems, it still has trouble
solving the instances corresponding to the two largest parameter
values. In contrast, \algo~is able to solve 
\emph{all} the problem instances at least an order of magnitude faster
than Gurobi. 
Note that the run time reported for Gurobi does not include the time
required to set up the LP. 
Note also that, when the sparsity parameter $\lambda = 0$, \algo\
is slower than Gurobi on  the smaller instances, but faster on the largest
instances; moreover, in contrast with Gurobi, \algo\ is able to solve \emph{all} the
instances in less than an hour. \algo\ is slower in this case because the
stopping criterion in subroutine FISTA~(see Algorithm~\ref{fig:FISTA}) is
harder to achieve when we do not regularize the portfolio by penalizing
its $\ell_1$-norm. We believe that changing the FISTA stopping criterion
to one better suited for the non-regularized problem, will significantly improve the
running time. 

The run times reported in Table~\ref{tab:results} are for the version
of \algo~that solves the constrained QP subproblems using an iterative line search.
In typical
applications, the portfolio selection problem is likely to have other
side constraints, and it is unlikely that one would be able to solve
the QP subproblems in this manner. In order to ensure that the
run times are not an artifact of the simple feasible set, we also
tested an implementation of \algo~where the QP step (and
also the gradient computation step) were solved using the quadratic
programming solver in Gurobi. The run times for this alternative
implementation were similar to those reported in
Table~\ref{tab:results}.

\begin{table}[tbp]
  \centering
  \footnotesize
  \begin{tabular}[h]{|c|rr|rrrrrrr|}
    \hline    
    \multirow{2}{*}{$\lambda$} & \multirow{2}{*}{$n$} & \multirow{2}{*}{$N$} & err &
    \multicolumn{2}{c}{Gurobi 5.0.2} & \multicolumn{2}{c}{Gurobi 5.6.0} &
    \multicolumn{2}{c|}{\algo} \\
    & & & (\%) & limit & time(s) & limit & time(s) & limit & time(s)\\
    \hline
    \multirow{11}{*}{$\lambda^\ast$} & 10 & 100 & 3.1 & -- & 0.01 & -- & 0.02 & -- & 0.12 \\
     & 10 & 500 & -- & -- & 0.51 & -- & 0.12 & -- & 0.24 \\
     & 10 & 1000 & 0.1 & -- & 1.54 & -- & 0.24 & -- & 0.29 \\
     & 10 & 1500 & -- & -- & 0.5 & -- & 0.48 & -- & 0.6 \\
     & 100 & 1000 & -- & -- & 4.61 & -- & 2.93 & -- & 3.21 \\
     & 100 & 5000 & 0.1 & -- & 230.37 & -- & 18.96 & -- & 14.19 \\
     & 100 & 10000 & -- & -- & 497.36 & -- & 54.33 & -- & 15.73 \\
     & 100 & 15000 & -- & -- & 98.38 & -- & 98.58 & -- & 67.7 \\
     & 1000 & 5000 & 0.1 & 1 & 943.61 & -- & 232.74 & -- & 63.6 \\
     & 1000 & 10000 & -- & 1 & 1050.15 & 1 & 1199.99 & -- & 247.44 \\
     & 1000 & 15000 & -- & 6 & 2538.93 & 5 & 2238.47 & -- & 440.07 \\
     \hline
     \multirow{11}{*}{$0$} & 10 & 100 & 0.2 & -- & 0.01 & -- & 0.01 & -- & 0.25 \\
     & 10 & 500 & 0.5 & -- & 0.27 & -- & 0.07 & -- & 0.28 \\
     & 10 & 1000 & -- & -- & 0.13 & -- & 0.13 & -- & 0.39 \\
     & 10 & 1500 & 0.2 & -- & 0.27 & -- & 0.22 & -- & 0.57 \\
     & 100 & 1000 & -- & -- & 1.61 & -- & 1.55 & -- & 13.76 \\
     & 100 & 5000 & -- & -- & 133.78 & -- & 10.61 & -- & 42.77 \\
     & 100 & 10000 & -- & -- & 26.02 & -- & 25.56 & -- & 94.97 \\
     & 100 & 15000 & -- & -- & 41.8 & -- & 40.88 & -- & 68.79 \\
     & 1000 & 5000 & -- & -- & 210.29 & -- & 93.45 & -- & 142.00 \\
     & 1000 & 10000 & -- & 9 & 3274.86 & -- & 286.17 & -- & 408.13 \\
     & 1000 & 15000 & -- & 9 & 3268.33 & 5 & 1960.7 & -- & 420.17 \\
    \hline
  \end{tabular}
  \caption{Average error (\emph{err})of \algo~with respect to Gurobi, number of 
  problems (out of 10) that reached a runtime limit of $1$ hour before finding 
  a solution and average run times of Gurobi and \algo~when solving random 
  instances of the spectral risk constrained portfolio optimization problem.}
  \label{tab:results}
\end{table}

\subsection{Parameter Uncertainty}

Next, we illustrate how the stability and scalability of \algo~can be used to 
overcome parameter uncertainty when hedging the risk of a portfolio of 
derivatives.

Suppose a portfolio manager wants to hedge the risk of an existing 
portfolio $\bx_0$ of derivative instruments using a set of $n$ liquid 
derivative 
positions. 
Let $\tilde{V}_0(\tilde{\bS}_t)$ and $\tilde{V}_i(\tilde{\bS}_t)$ denote, 
respectively, the
value of the initial portfolio $\bx_0$ and the value of derivative
instrument~$i \in \{1, \ldots, n\}$ at time $t$, as
functions of  
the 
vector of 
underlying asset prices 
$\tilde{\bS}_t\in \mR^s$. 
Let $\tilde{\ell}_0(t) = \tilde{V}_0(\tilde{\bS}_0) - 
\tilde{V}_0(\tilde{\bS}_t)$ (resp. $\tilde{\ell}_i(t) = \tilde{V}_i(\tilde{\bS}_0) - 
\tilde{V}_i(\tilde{\bS}_t)$) denote the loss of the 
initial portfolio (resp. derivative instrument $i$) at time~$t$.
Then, the loss at time $t$ of a hedging portfolio $\bx \in \mR^n$ is given by 
$\sum_{i=1}^n \tilde{\ell}_i(t) x_i$, and the total loss at time $t$ for the 
portfolio manager is $\tilde{\ell}_0(t) + \sum_{i=1}^n \tilde{\ell}_i(t) x_i$.
Note that, in contrast with our previous notation, $x_i$ now denotes the 
total \emph{number} of units of derivative $i$ purchased. Therefore, in
what follows we drop the portfolio constraint  
$\ones^\top \bx = 1$.

Suppose the
underlying asset prices $\tilde{\bS}_t$ are log-normally distributed with mean 
vector $\bfpi$ and \emph{unknown} covariance matrix $\tilde{\bfSigma}_t = 
\tilde{\bD}_t \bR 
\tilde{\bD}_t$, where $\bR$ is a constant correlation matrix and 
$\tilde{\bD}_t = \text{diag}(\tilde{\bfsigma}_t)$ is a diagonal matrix 
of \emph{unknown} volatilities  at time $t$.
Suppose the portfolio manager knows the current volatility $\bfsigma_0$, and 
believes that the volatility at the time horizon $T$ is of the form $\bfsigma_T 
= 
\bfsigma_0 + \sum_{p=1}^q \omega_p \bfrho_p$, where $\bfrho_p \in \reals^s$ are 
known factors and $\omega_p \in [-1, 1]$ are the corresponding \emph{unknown}
weights.
For $\bfomega \in \Omega : = \{-1, 1\}^q \cup \{\bo\}$, let 
$\bfell_0(\bfomega) \in \reals^N$  (resp $\bfell_i(\bfomega) \in \reals^N$) 
denote the vector of $N$ samples of the loss $\tilde{\ell}_0(T)$ on the initial 
portfolio  (resp. the loss $\tilde{\ell}_i(T)$  on derivative
instrument~$i$) when the 
volatility vector $\bfsigma_T = \bfsigma_0 + \sum_{p=1}^q \omega_p \bfrho_p$.
For a subset $W$ of $\Omega$, consider the following hedging portfolio 
selection problem:
\begin{align}
\Pi(W) : &=
\begin{array}[t]{rll}
\max\limits_{\bx} & \min_{\bfomega \in W} \left\{ \mu_0 + 
\bfmu(\bfomega)^{\top}\bx - 
\lambda 
\norm{\bx}_{1} \right\} \\
\text{s.t.} & \ES{\beta}{\bfell_0(\bfomega) + \bL(\bfomega)\bx} \leq \alpha 
\ES{\beta}{\bfell_0(\bfomega)}, & \bfomega \in 
W\\
& \norm{\bx}_\infty \leq B\\
\end{array}
\label{eq:derivatives} \\
&= \begin{array}[t]{rll}
\max\limits_{\bx,\mu} & \mu - \lambda 
\norm{\bx}_{1} \\
\text{s.t.} & \mu \leq \mu_0 + \bfmu(\bfomega)^\top \bx, & \bfomega \in W\\
&\ES{\beta}{\bfell_0(\bfomega) + \bL(\bfomega)\bx} \leq \alpha 
\ES{\beta}{\bfell_0(\bfomega)}, & \bfomega \in W\\
& \norm{\bx}_\infty \leq B,\\
\end{array} \nonumber
\end{align}
where $\bL(\bfomega) =
[\boldsymbol{\ell}_{1}(\bfomega) \ldots
\boldsymbol{\ell}_n(\bfomega)]$,  
 $\bfmu(\bfomega) = -\frac{1}{N} \bL(\bfomega)^\top \ones$,
 and $\mu_0 = -\frac{1}{N} \bfell_0(\bfomega)^\top \ones$. By solving 
 problem~\eqref{eq:derivatives}, the portfolio 
 manager is looking to compute an $\ell_1$-regularized 
 hedging portfolio $\bx$ that maximizes the worst-case (w.r.t. $W$) 
 expected return of the total portfolio $[\bx_0^\top, \bx^\top]^\top$,
 while ensuring that the worst case expected shortfall drops by 
 factor of $\alpha < 1$.
We define $\Pi(\{\bo\})$ (resp. $\Pi(\{-1,1\}^q)$) as the nominal (resp. 
robust) portfolio selection problem.
Since we allow the hedging portfolio $\bx$ to have both long and short 
positions, 
in order to be robust against uncertainty in the parameters
$\omega_p$ we must consider all the possible worst-case risk models $\bfomega 
\in 
\{-1,1\}^q$.
Problem~\eqref{eq:derivatives} is equivalent to
\begin{align}
  \begin{array}[t]{rll}
  \max\limits_{\bar{\bx}} & \bar{\bfmu}^\top \bar{\bx} -
  \lambda \norm{\bar{\bx}}_{1} \\
  \text{s.t.} & \ES{0}{\bfell_0(\bfomega) + \hat{\bL}(\bfomega)\bar{\bx}} 
    \leq 0, & \bfomega \in W\\
  & \ES{\beta}{\bfell_0(\bfomega) + \bar{\bL}(\bfomega)\bar{\bx}} 
  \leq \alpha 
  \ES{\beta}{\bfell_0(\bfomega)}, & \bfomega \in 
  W\\
  & \bl \leq \bar{\bx} \leq \bu,\\
  \end{array}
\end{align}
where $\bar{\bx} = [\bx^\top, \mu^+, \mu^-]^\top$, 
$\bar{\bfmu} = [\bo^\top, \lambda+1, \lambda-1]^\top$, 
$\hat{\bL}(\bfomega) = [\bL(\bfomega), \ones, -\ones]$, 
$\bar{\bL}(\bfomega) = [\bL(\bfomega), \bo, \bo]$, 
$\bl = [-B \ones^\top, 0, 0]^\top$, and 
$\bu = [B\ones^\top, \infty, \infty]^\top$. Thus, by slightly modifying 
\algo~to deal with box constraints of the form $\bl \leq \bx \leq \bu$ instead 
of the portfolio and leverage constraints $\ones^\top \bx = 1$ and 
$\norm{\bx}_\infty \leq B$, we are able to exploit its stability and 
scalability 
% --illustrated in Section~\ref{sec:scale}--
to construct hedging portfolios that 
are 
robust against parameter uncertainty.
%We would like to modify nominal problem (\ref{eq:nominal}) in order to produce
%optimal portfolios that will not violate the risk constraint once the real
%parameters $\bfpi$ and $\bfSigma$, and hence the asset prices
%$\tilde{\bS}_{T}$ and the derivative values $\tilde{\bV}(T)$,
%are realized. Suppose that $\bfpi = \hat{\bfpi}$
%but $\bfSigma = U_{5\%} \hat{\bfSigma}$, where $U_{5\%}$ is a uniform random
%variable in the interval $[0.95, 1.05]$. 
%Note that both long and short positions are allowed in portfolio $\bx$.
%Consider the ``robust'' sparse
%derivatives portfolio selection problem, which includes the two worst case risk models corresponding to $k=0.95$ and $k=1.05$: 
%\begin{equation}
%\begin{array}{rll}
%\max & \bfmu^{\top}\bx - \lambda \norm{\bx}_{1}\\
%\text{s.t.} & \ES{\beta}{\bL_k\bx} \leq \alpha, & k \in \{1-\kappa, 1, 1+\kappa\},\\
%& \bV(0)^{\top}\bx =1, \\
%& \norm{\bx}_\infty \leq B.\\
%\end{array}
%\label{eq:robust}
%\end{equation}

In what follows, we show that, using \algo, one can
construct a portfolio that reduces the risk of the initial portfolio while 
removing the impact of the uncertain parameters on the expected 
return.
Following~\cite{alexander2003derivative}, we assumed that the initial 
portfolio 
consisted of four short positions of European at-the-money binary call options, 
each on one of four correlated assets, with maturity in 4, 6, 8, and 10 months, 
respectively. The hedging universe was composed of 20 vanilla European calls on 
each asset, given by the combination of strike prices $[0.9, 0.95, 1, 1.05, 
1.1] S_0$ and maturities $[2, 3, 4, 6]$ months, and the assets themselves. The 
time horizon was $T = 1$ month. We used $N=25,000$ Monte Carlo samples to 
simulate the underlying asset prices. The derivatives were priced using 
Black-Scholes formulae. The rest of the problem parameters were set as follows: 
the $q=2$ factors affecting the volatility, $\bfrho_1 = 
0.02[1,1,1,1]^\top$ 
and $\bfrho_2 = 0.02[1, -1, 1, -1]^\top$; the expected shortfall level 
$\beta=0.95$; 
the risk reduction factor $\alpha = 0.5$, i.e. the portfolio manager is looking 
reduce his exposure by half; the leverage bound 
$B=1$; and the parameter controlling the sparsity of the portfolio $\lambda = 
\theta \frac{2\bfmu(\bo)^\top  \bx^\ast}{\norm{\bx^\ast}_1}$, where $\theta \in
\{0, 0.5, 1\}$, and $\bx^\ast$ is the optimal solution to $\Pi(\{\bo\})$ with
$\lambda = 0$.
  
%Figure~\ref{fig:norms} shows the $\ell_0$-norm, the $\ell_1$-norm, and the 
%mean return of the optimal nominal and robust portfolios as functions of 
%parameter $\theta$ (constant $\alpha = 1$), and of the risk budget $\alpha$ 
%(constant $\theta = 8$). As $\theta$ increases, the $\ell_1$-norm penalty 
%$\lambda$ in~\eqref{eq:derivatives} decreases; therefore, the $\ell_1$-norm 
%and 
%the mean return of the optimal portfolios increase. Similarly, the 
%$\ell_1$-norm and the mean return of the optimal portfolios are increasing 
%functions of the risk budget $\alpha$. The support or $\ell_0$-norm of the 
%robust portfolio is a monotonic function of $\theta$ and $\alpha$, but this 
%does not seem to be the case for the nominal portfolio. Recall that we use the 
%$\ell_1$-norm as a convex proxy for the $\ell_0$-norm.
%Note that the nominal and robust portfolios have a similar $\ell_1$-norm and 
%mean return.
%These values are also similar for the super-robust portfolio, but where omitted from Figure \ref{fig:norms} for ease of visualization.

%\begin{figure}
%	\begin{center}
%		\includegraphics[width=\textwidth]{norms.pdf}
%		\caption{$\ell_0$-norm, $\ell_1$-norm, and mean return of the nominal and 
%robust
%			portfolios as functions of the parameter controlling the sparsity of the 
%portfolio $\theta$ (left) and the risk budget $\alpha$ (right).}
%		\label{fig:norms}
%	\end{center}
%\end{figure}

Figures~\ref{fig:riskret0} and~\ref{fig:riskret1} show the out-of-sample 
expected shortfall and mean return of the initial,
nominal and robust portfolios, as functions of the uncertain parameters 
$(\omega_1, \omega_2) \in [-1,1]\times\{-1,0,1\}$, when the sparsity parameter 
$\theta = 0$ and 
$\theta = 1$, respectively.
Note that, in all cases, the risk constraint $\ES{\beta}{\bfell_0(\bfomega) + 
\bar{\bL}(\bfomega)\bar{\bx}} 
  \leq \alpha \ES{\beta}{\bfell_0(\bfomega)}$ is violated by the nominal 
  portfolio
  for $\omega_1 > 0$. On the other hand, the risk of the final robust portfolio 
  is always less than half of that of the initial portfolio, regardless of the 
  uncertain parameter values. In addition, the expected rate of return of the 
  robust portfolio is virtually independent of the uncertain parameters 
  $(\omega_1, \omega_2)$. In contrast, the expected rate of return of the 
  nominal portfolio varies significantly as the uncertain parameters $\omega_1$ 
  and $\omega_2$ change. 
%  Thus, we are able to construct robust portfolios with almost constant mean 
%  rate return, and guaranteed to decrease the exposure of the portfolio 
%manager 
%  by at least half for every realization of the uncertain parameters.
  Note that 
  we are able to solve for the robust portfolio only
  because \algo\ is computationally much more efficient as compared to the 
  naive LP approach. In fact, 
  \algo\ is so efficient that  one can solve portfolio selection problems with 
  more complicated uncertainty in the covariance matrix $\bfSigma$, or 
  uncertainty in the mean return vector $\bfpi$, by including more risk 
  constraints in~\eqref{eq:derivatives}. Finally, Figure~\ref{fig:holdings} 
  shows the positions $x_i$ of the optimal 
  nominal and robust
  portfolios, for $\theta=0.5$ and $\theta=1$. Note that the robust porfolio 
  holds position in almost all the instruments that the nominal porftolio does. 
  However, the robust portfolio holds positions in other additional assets. 
  These positions have the desired effect of reducing the out-of-sample risk 
  and reducing the expected return variance. It is also worth noting that the 
  sparsity parameter $\theta$ seems to have a larger impact on the robust 
  portfolio holdings than on the nominal portfolio ones.

\begin{figure}
	\begin{center}
		\includegraphics[width=\textwidth]{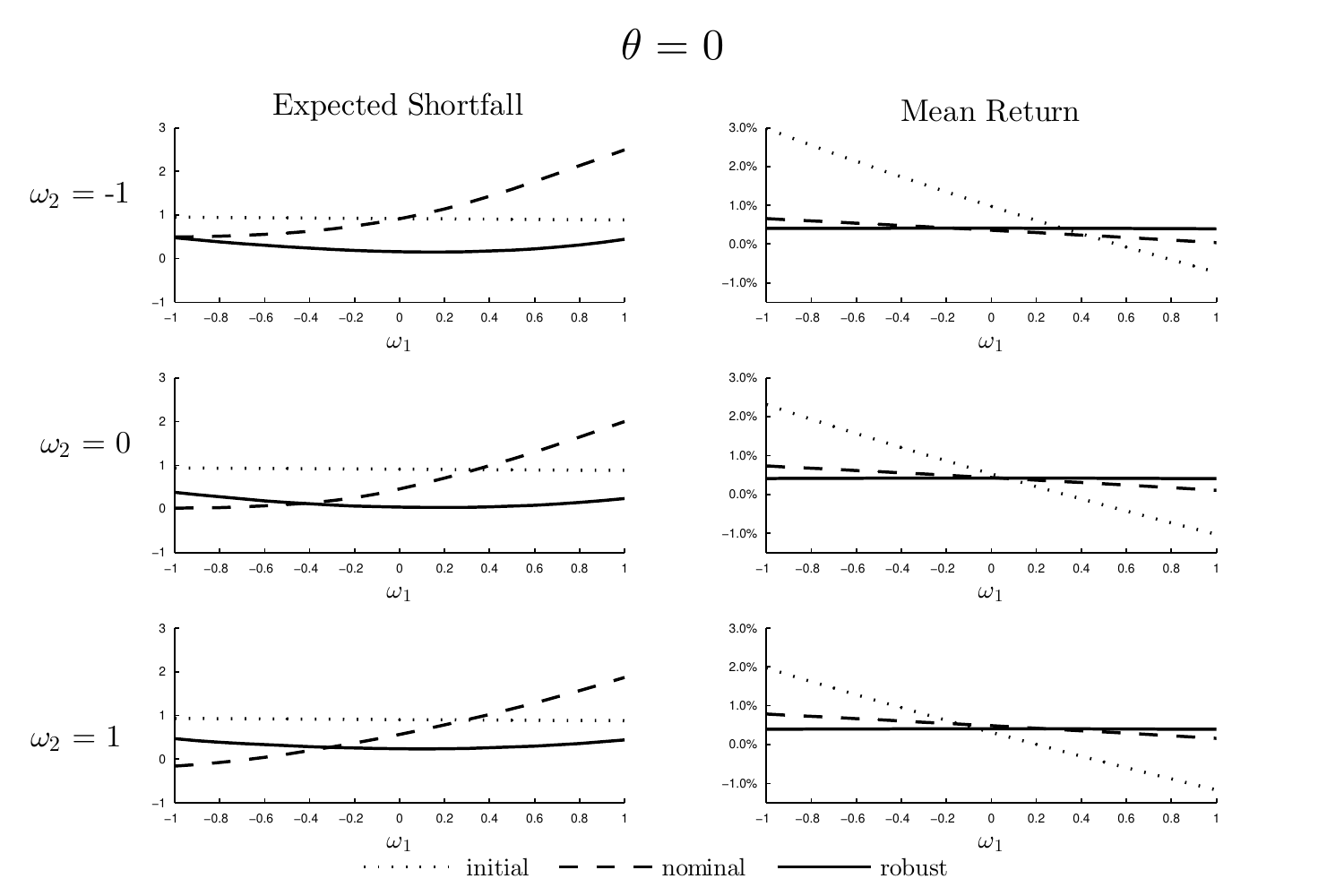}
		\caption{Out-of-sample expected shortfall and mean return of the initial, 
		nominal and robust portfolios, as a function of the uncertain parameters 
		$(\omega_1, \omega_2) \in [-1, 1]\times\{-1,0,1\}$. The sparsity parameter 
		$\theta = 0$.}
		\label{fig:riskret0}
	\end{center}
\end{figure}
\begin{figure}
	\begin{center}
		\includegraphics[width=\textwidth]{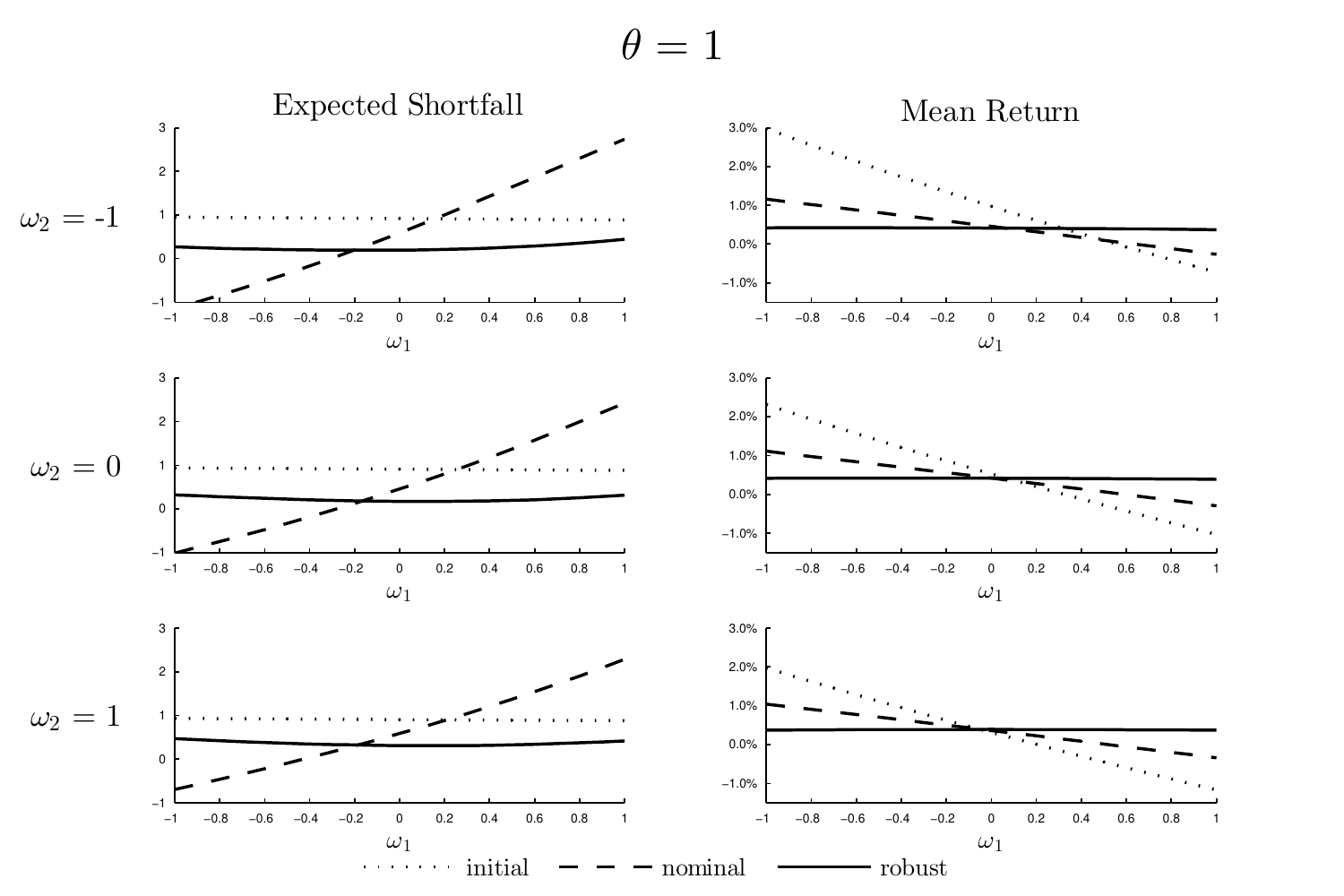}
		\caption{Out-of-sample expected shortfall and mean return of the initial, 
		nominal and robust portfolios, as a function of the uncertain parameters 
		$(\omega_1, \omega_2) \in [-1, 1]\times\{-1,0,1\}$. The sparsity parameter 
		$\theta = 1$.}
		\label{fig:riskret1}
	\end{center}
\end{figure}

\begin{figure}
	\begin{center}
		\includegraphics[width=\textwidth]{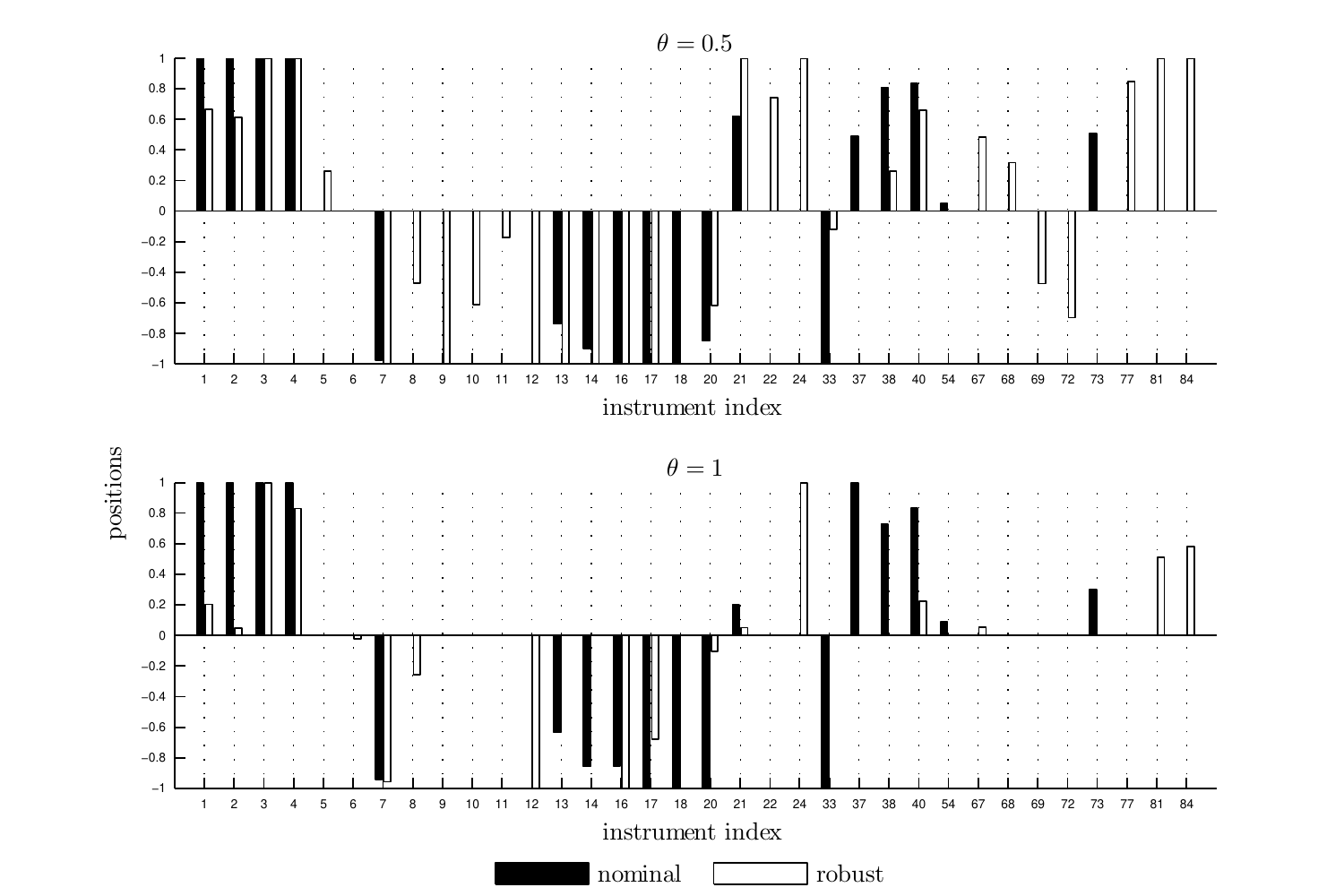}
		\caption{Holdings of the nominal and robust portfolios. The sparsity 
		parameter $\theta=0.5$ (top) and $\theta=1$ (bottom).}
		\label{fig:holdings}
	\end{center}
\end{figure}

\section{Conclusion}\label{sec:conc}
In this paper, we propose a simple gradient-based
algorithm \textsc{SpecRiskAllocate} for solving the portfolio selection
problem with multiple spectral risk constraints. 
This algorithm computes the optimal portfolio by solving a
sequence of separable convex QPs over the initial feasible set,
i.e. the formulation does not increase the dimension of the problem to
represent the risk measures. \textsc{SpecRiskAllocate} is very
efficient both in theory and in practice. Our numerical experiments
show that \textsc{SpecRiskAllocate} is at least one order of magnitude
faster than the state-of-the-art general purpose solver on most
instances of the spectral risk constrained portfolio selection problem that
are of practical interest. Moreover, our numerical experiments show
that \textsc{SpecRiskAllocate} allows portfolio managers to impose
constraints with respect to multiple risk models as a means of inducing
robustness in their portfolios against parameter uncertainty. 

\newpage
\bibliography{references}

\newpage
\appendix
\section{Smoothing of $g(x)$}
\label{app:smoothing}
Define the function
\begin{equation}
  \label{eq:fnu-def}
  f_{\beta}^{(\nu)}(\bfzeta) =
  \begin{array}[t]{rl}
    \max & \bfzeta^{\top}\bq - \frac{\nu}{2}\norm{\bq}_{2} \\
    \text{s.t.} &  0 \leq \bq \leq \frac{1}{(1-\beta)N}\ones, \\
    & \ones^{\top} \bq = 1.
  \end{array}
\end{equation}
\cite{Nesterov2005} establishes
that $f_{\beta}^{(\nu)}(\bfzeta)$ is a differentiable strongly convex
function with gradient $\nabla f^{(\nu)}_{\beta}(\bfzeta) = \bq^\ast$, where 
$\bq^\ast$ is
the unique solution to~\eqref{eq:fnu-def}. The gradient $\nabla f^{(\nu)}_{\beta}$ is Lipschitz
continuous with Lipschitz constant $\sfrac{1}{\nu}$.  Moreover,
$f_{\beta}^{(\nu)}$ satisfies $\ES{\beta}{\bfzeta} - \nu \leq
f^{(\nu)}_{\beta}(\bfzeta) \leq \ES{\beta}{\bfzeta}$,
i.e. $f^{(\nu)}_{\beta}(\bfzeta)$ is a $\nu$-approximation to
$\ES{\beta}{\bfzeta}$. 

Let $\rho(\bfzeta) = \sum_{\ell=1}^d \gamma_\ell
\ES{\beta_\ell}{\bfzeta}$ denote any generalized spectral risk
function. We define the smoothed spectral risk function as
\eq
\rho^{(\nu)}(\bfzeta) = \sum_{\ell=1}^d \gamma_\ell
f^{(\nu)}_{\beta_\ell}(\bfzeta).
\en
Since $\sum_{\ell=1}^{d} \gamma_{\ell} = 1$ for all generalized spectral risk functions, it
follows that $\rho(\bfzeta) - \nu \leq \rho^{(\nu)}(\bfzeta) \leq
\rho(\bfzeta)$. The gradient of $\rho^{(\nu)}(\bfzeta)$ is given by
$\nabla \rho^{(\nu)}(\bfzeta) = \sum_{\ell=1}^d \gamma_{\ell} \bq^\ast_{\ell}$,
where $\bq^\ast_{\ell}$ is the unique optimal solution to
(\ref{eq:fnu-def}) with $\beta = \beta_{\ell}$.

Finally, define
\begin{equation}
  \label{eq:psi-def}
  \Psi^{(\delta)}(\bt) = 
  \begin{array}[t]{rl}
    \max & \bt^{\top}\bu - \frac{\delta}{2}\norm{\bu}_{2} \\
    \text{s.t.} & \ones^{\top}\bu = 1 \\
    & \bu \geq \mathbf{0}.\\
  \end{array}
\end{equation}
$\Psi^{(\delta)}$ is a differentiable convex function with Lipschitz continuous gradient $\nabla \Psi^{(\delta)}(\bt) = \bu^\ast$, where $\bu^\ast$ is the unique solution to \eqref{eq:psi-def}, and Lipschitz constant
$\sfrac{1}{\delta}$~\citep{Nesterov2005}. In addition, we have that
$\Psi(\bt) - \delta \leq \Psi^{(\delta)}(\bt) \leq \Psi(\bt)$.

We define the smoothing of $g(\bx)$ as
\eq g_{\nu\delta}(\bx) =
\Psi^{(\delta)}\left(\rho^{(\nu)}_1(\bL_1\bx)-\alpha_1, \ldots,
  \rho^{(\nu)}_m(\bL_m\bx)-\alpha_m, 0\right).  \en
Theorem [7] in~\cite{IyengarPS11:packing} (see,
also~\cite{HodaGilpinPena2010}) guarantees that $g_{\nu\delta}(\bx)$
is a convex function with Lipschitz continuous gradient
\begin{equation}
	\nabla g_{\nu\delta}(\bx) = \sum_{k=1}^m u^\ast_k \bL_k^\top \nabla\rho^{(\nu)}_k(\bL_k \bx),
	\label{eq:gradient-gsmooth}
\end{equation}
where \eq \bu^\ast = \begin{array}[t]{rl}
		\argmax & \sum_{k=1}^m u_k \left( \rho^{(\nu)}_k(\bL_k\bx)-\alpha_k \right) -  \frac{\delta}{2}\norm{\bu}_{2} \\
		\text{s.t.} & \ones^{\top}\bu = 1 \\
					& \bu \geq \bo.
					\end{array} \en
Moreover, $g_{\nu\delta}(\bx)$ is a $(\nu+\delta)$-approximation to $g(\bx)$, i.e. $g(\bx) - \nu - \delta \leq g_{\nu\delta}(\bx) \leq g(\bx)$.

\section{Details of \textsc{SpecRiskAllocate}}
\label{app:specriskalgo}

Recall that
the \FISTA~iterates are computed by solving an $\ell_1$-penalized
QP
of the form \eqref{eq:FISTA-mv}.
Next, we show how to solve this
problem using a one-dimensional search.
Dualizing the constraint $\ones^\top\bx = 1$, we obtain the following optimization problem:
$$\cL(\gamma) = \min_{\norm{\bx}_{\infty} \leq B} \left\{ \eta\lambda\norm{\bx}_1 + (\bfxi - C\by + \gamma\ones)^\top \bx + \frac{C}{2} \bx^\top\bx \right\}.$$
Writing $\bx = \bw- \bv$, where $\bw ,\bv \geq \bo$, observe that
\begin{align*}
	\cL(\gamma) = & \min_{\bo \leq \bw \leq B\ones} \left\{ \left( \eta\lambda\ones + \bfxi -
    C\by+\gamma\ones \right)^\top \bw + \frac{C}{2} \bw^\top\bw \right\} \\
		& + \min_{\bo \leq \bv \leq B\ones} \left\{ \left( \eta\lambda\ones - \bfxi + C\by - \gamma\ones \right)^\top \bv 
  - \frac{C}{2} \bv^\top\bv \right\},
\end{align*}
where we have ignored the cross terms $\bw^\top\bv$ because they are zero in any optimal solution. The optimal solution to $\cL(\gamma)$ is given by
$x^\ast_i(\gamma) = \min\left\{(\bar{c}_i -\gamma)/C,B\right\}^+ - 
  \min\left\{(\underline{c}_i +\gamma)/C,B\right\}^+$,
where $\bar{c}_i = -\eta\lambda - \xi_i + C y_i$, and $\underline{c}_i = -
\eta\lambda + \xi_i - C y_i$, $i = 1,\ldots, n$. The optimal solution
to~\eqref{eq:FISTA-mv} can
be recovered by finding the dual variable $\gamma^\ast$ such that
$\ones^\top\bx^\ast(\gamma^\ast) = 1$. Since $\lim_{\gamma \rightarrow
  \infty}\bx^\ast(\gamma) = -B\ones$ and $\lim_{\gamma \rightarrow
  -\infty} \bx^\ast(\gamma) = B\ones$, it follows that there exists
$\gamma^\ast \in (-\infty,\infty) $ such that $\ones^\top\bx^\ast(\gamma^\ast) =
1$. The computational complexity of finding $\gamma^\ast$ is
dominated by the computational cost of sorting the set $\cup_{1\leq i
  \leq n}\{\bar{c}_i, \underline{c}_i\}$.

\FISTA~(see Algorithm \ref{fig:FISTA}) calls subroutine \text{ComputeGradient}, displayed in Algorithm~\ref{fig:grad}, to compute the gradient $\bfxi$.
Computing gradient $\bfxi$ requires computing the gradient $\nabla g_{\nu\delta}(\bx)$  (cf.~\eqref{eq:gradient-gsmooth}), which requires solving one QP of the form
(\ref{eq:psi-def}) and $\sum_{k=1}^m d_k$ QPs of the form~\eqref{eq:fnu-def}. Each of
these QPs is of the form
\begin{equation}
  \label{eq:grad-qp}
  \begin{array}{rl}
    \mbox{max} & \bc^\top\bx - \frac{1}{2}\norm{\bx}_2^2,\\
    \mbox{s.t.} & \ones^\top\bx = 1, \\
    & \bo \leq \bx \leq \bb,
  \end{array}
\end{equation}
where the bound $\bb \geq 0$ satisfies $\ones^\top \bb\geq 1$, and is
possibly infinite. Dualizing the constraint $\ones^\top\bx = 1$, we obtain the following separable QP:
\eq \cL(\gamma) = \max_{0 \leq \bx \leq \bb} \left\{ \sum_{i=1}^n(c_i-\gamma)x_i - \frac{1}{2}x_i^2 \right\}.\en
The optimal solution to $\cL(\gamma)$ is given by $x^\ast_i(\gamma) =
\min\{c_i-\gamma,b_i\}^+$, $i = 1,\ldots, n$. The
optimal solution to~\eqref{eq:grad-qp} can be recovered by finding the
dual variable $\gamma^\ast$ such that $\ones^\top\bx^\ast(\gamma^\ast) = 1$. Since
$\lim_{\gamma \rightarrow \infty}\bx^\ast(\gamma) = \bo$ and
$\lim_{\gamma \rightarrow -\infty} \bx^\ast(\gamma) = \bb$, it follows
that there exists $\gamma^\ast \in (-\infty,\infty) $ such that
$\ones^\top\bx^\ast(\gamma^\ast) = 1$. The computational complexity of
computing $\gamma^\ast$ is dominated by the computational cost of sorting
the set $\cup_{1 \leq i \leq n}\{c_i,c_i-b_i\}$.

\begin{algorithm}[tp]
  \begin{algorithmic}[1]
    \FOR {$k=1$ \TO $m$} \FOR{$\ell=1$ \TO $d_{k}$} \STATE
    $\bq_{k\ell} \gets \argmax \left\{\bq^{\top}\bL_{k}\by -
      \frac{\nu}{2}\norm{\bq}_{2}^{2} : \ones^{\top} \bq = 1, 0\leq
      \bq \leq \frac{1}{(1-\beta_{k\ell})N_k}\ones\right\}$
    \ENDFOR
    \ENDFOR
    \STATE $\bu \gets \argmax
    \left\{\sum_{k=1}^{m}v_{k}(\rho_{k}^{(\nu)}(\bL_k\by)-\alpha_{k}) -
      \frac{\delta}{2}\norm{\bv}_{2}^{2} : \sum_{k=1}^{m+1} v_k = 1,
      \bv\geq \bo\right\}$ \STATE $\bfxi \gets -\eta\bfmu +
    \sum_{k=1}^{m} u_{k}\left(\sum_{\ell=1}^{d_{k}}
      \gamma_{k\ell}\bL_{k}^{\top} \bq_{k\ell}\right)$ \RETURN $\bfxi$
  \end{algorithmic}
  \caption{Function \textsc{ComputeGradient}($\by$, $\nu$, $\delta$)}
  \label{fig:grad}
\end{algorithm}

\end{document}